\documentclass[a4paper]{article}
\usepackage[dvipsnames]{xcolor}
\usepackage{url}
\usepackage{a4wide}
\usepackage{amsmath,amsfonts}

\usepackage{xspace}
\usepackage{microtype}
\usepackage[vlined,linesnumbered,ruled,longend]{algorithm2e}

\usepackage{float}
\makeatletter
\newcommand\floatc@mybox[2]{\vbox{\hbadness10000
\moveleft3.4pt\vbox{\advance\hsize by6.8pt
\hrule \hbox to\hsize{\vrule\kern3pt
\vbox{\kern3pt\vbox{\advance\hsize by-6.8pt{\@fs@cfont #1} #2}\kern3pt}\kern3pt\vrule}}}}%
\newcommand\fs@mybox{\def\@fs@cfont{\bfseries}\let\@fs@capt\floatc@mybox
\def\@fs@pre{\setbox\@currbox\vbox{\hbadness10000
\moveleft3.4pt\vbox{\advance\hsize by6.8pt
\hrule \hbox to\hsize{\vrule\kern3pt
\vbox{\kern4.5pt\box\@currbox\kern4.5pt}\kern3pt\vrule}\hrule}}}%
\def\@fs@mid{}%
\def\@fs@post{}%
\let\@fs@iftopcapt\iftrue}
\makeatother
\floatstyle{mybox}
\newfloat{inset}{tb}{ins}
\floatname{inset}{Inset}

\usepackage{graphicx}
\usepackage{wrapfig}
\usepackage{amsthm}
\usepackage{tabularx}
\usepackage[normalem]{ulem}

\theoremstyle{plain}
\newtheorem{theorem}{Theorem}
\newtheorem{definition}[theorem]{Definition}
\newtheorem{lemma}[theorem]{Lemma}
\newtheorem{corollary}[theorem]{Corollary}
\newtheorem{observation}[theorem]{Observation}

\newcommand{\XSays}[3]{{\color{#2}
      {$\rule[-0.12cm]{0.2in}{0.5cm}$\fbox{\tt
            #1:} }
       #3
      \def\comment{#3}\def\empty{}\ifx\comment\empty\else
      {
          \fbox{\tt end}
      } \fi
      }%
}

\newcommand\curve{\ensuremath{f}\xspace}
\newcommand\BCR{\textsc{bcr}\xspace}
\newcommand\unsigneddim{\ensuremath{\{1,\ldots,d\}}\xspace}
\newcommand\signeddim{\ensuremath{\{-d,-d+1,\ldots,d-1,d\}\setminus\{0\}}\xspace}
\newcommand\twodk{\ensuremath{2^{d\ast k}}\xspace}
\newcommand\reverse[1]{\ensuremath{\overleftarrow{#1}}}
\newcommand\perm{\ensuremath{\sigma}\xspace}

\DeclareMathOperator{\sgn}{sign}
\DeclareMathOperator{\flipped}{isneg}
\DeclareMathOperator{\permdp}{depth}
\DeclareMathOperator{\edgedist}{ed}
\DeclareMathOperator{\localedgedist}{led}
\DeclareMathOperator{\absentr}{ent}
\DeclareMathOperator{\absexit}{ext}
\DeclareMathOperator{\relentr}{rlent}
\DeclareMathOperator{\relexit}{rlext}
\DeclareMathOperator{\compose}{\ensuremath{\circ}}
\DeclareMathOperator{\volm}{vol}
\newcommand\volume[1]{\ensuremath{\volm(#1)}}

\title{\MakeUppercase{Hyperorthogonal well-folded Hilbert curves}}

\author{
  Arie Bos%
  \thanks{Dept. of Mathematics and Computer Science; Eindhoven University of Technology, the Netherlands,
          arie\_bos@online.nl}\,
  and
  Herman~Haverkort%
  \thanks{Dept. of Mathematics and Computer Science; Eindhoven University of Technology, the Netherlands,
          cs.herman@haverkort.net}}

\begin{document}
\maketitle

\begin{abstract}
R-trees can be used to store and query sets of point data in two or more dimensions. An easy way to construct and maintain R-trees for two-dimensional points, due to Kamel and Faloutsos, is to keep the points in the order in which they appear along the Hilbert curve. The R-tree will then store bounding boxes of points along contiguous sections of the curve, and the efficiency of the R-tree depends on the size of the bounding boxes---smaller is better. Since there are many different ways to generalize the Hilbert curve to higher dimensions, this raises the question which generalization results in the smallest bounding boxes. Familiar methods, such as the one by Butz, can result in curve sections whose bounding boxes are a factor $\Omega(2^{d/2})$ larger than the volume traversed by that section of the curve. Most of the volume bounded by such bounding boxes would not contain any data points. In this paper we present a new way of generalizing Hilbert's curve to higher dimensions, which results in much tighter bounding boxes: they have at most 4 times the volume of the part of the curve covered, independent of the number of dimensions. Moreover, we prove that a factor 4 is asymptotically optimal.
\end{abstract}

\tableofcontents

\clearpage

\section{Introduction}\label{sec:intro}

\subsection{Space-filling curves and spatial index structures}\label{sec:rtrees}

A $d$-dimensional space-filling curve is a continuous, surjective mapping from $\mathbb{R}$ to $\mathbb{R}^d$.
In the late 19th century Peano~\cite{Peano} described such mappings for $d = 2$ and $d = 3$. Since then, various other space-filling curves have been found, and they have been applied in diverse areas such as spatial databases, load balancing in parallel computing, improving cache utilization in computations on large matrices, finite element methods, image compression, and combinatorial optimization~\cite{Bader,Haverkort2009,Sagan}.
In this paper we present new space-filling curves for $d > 2$ that have favourable properties for use in spatial data structures.

In particular, we consider data structures for $d$-dimensional points such as R-trees~\cite{Manolopoulos}. In such data structures, data points are organised in blocks, often stored in external memory. Each block contains at most $B$ points, for some parameter $B$, and each point is stored in exactly one block. For each block we maintain a bounding box, which is the smallest axis-aligned $d$-dimensional box that contains all points stored in the block. The bounding boxes of the blocks are stored in an index structure, which may often be kept in main memory. To find all points intersecting a given query window $Q$, we can now query the index structure for all bounding boxes that intersect $Q$; then we retrieve the corresponding blocks, and check the points in those blocks for answers to our query. We may also use the index structure to find the nearest neighbour to a query point~$q$: if we search blocks in order of increasing distance from~$q$, we will retrieve exactly the blocks whose bounding boxes intersect the largest empty sphere around~$q$. The grouping of points into blocks determines what block bounding boxes are stored in the index structure, and in practice, retrieving these blocks is what determines the query response time~\cite{Haverkort2009}.

If we store $n$ points in $d$ dimensions with $B$ points in a block, $\Theta((n/B)^{1-1/d})$ blocks may need to be visited in the worst case if the query window is a rectangular box with no points inside~\cite{kanth}, and $\Theta(n/B)$ blocks may need to be visited if the query window is an empty sphere. The \emph{Priority-R-tree} achieves these bounds~\cite{prtree}, whereas a heuristic solution by Kamel and Faloutsos~\cite{Kamel1993}, which is explained below, may result in visiting $\Theta(n/B)$ blocks even if the query window is a rectangular box with no points inside~\cite{prtree}. However, experimental results for (near-)point data and query ranges with few points inside~\cite{alenex} indicate that the approach by Kamel and Faloutsos seems to be more effective in practice for such settings. Moreover, regardless of the type of data and query ranges, a structure based on the ideas of Kamel and Faloutsos is much easier to build and maintain than a Priority-R-tree~\cite{prtree}.

\begin{figure}[b]
  \centering
  \leavevmode
  \raisebox{27mm}{(a)} \includegraphics[height=30mm]{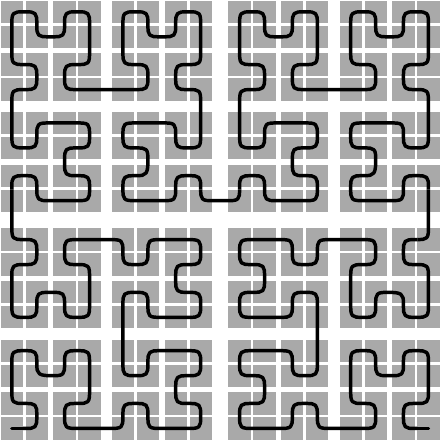}
  \hfill
  \raisebox{27mm}{(b)} \includegraphics[height=30mm]{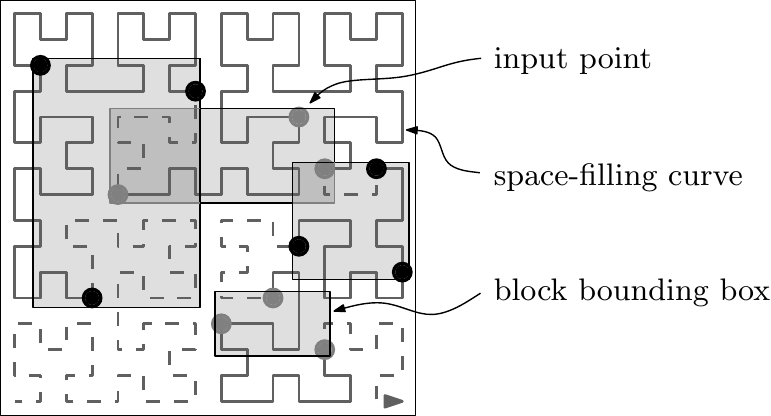}
  \hfill
  \raisebox{27mm}{(c)} \includegraphics[height=30mm]{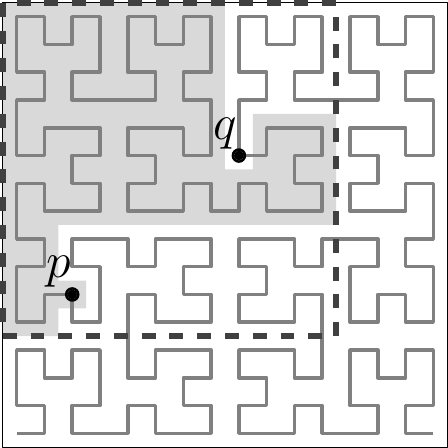}
  \caption{%
  (a) Sketch of Hilbert's space-filling curve.\quad
  (b) Blocks of an R-tree or similar data structure with $B = 3$.\quad
  (c) Box-to-curve ratio of the section between $p$ and $q$ =
  area of the bounding box of the curve section $S$ between $p$ and $q$, divided by the area covered by $S$: $12\ast 12 / 87 \approx 1.66$.}
  \label{fig:rtree_leaves}
  \label{fig:measures}
  \label{fig:hilbert2d}
\vspace{-0.5cm}
\end{figure}

Kamel and Faloutsos proposed to determine the grouping of points into blocks as follows: we order the input points along a space-filling curve and then put each next group of $B$ points together in a block (see Figure~\ref{fig:rtree_leaves}(b)). Note that the number of blocks retrieved to answer a query is simply the number of bounding boxes intersected. Therefore it is important that the ordering induced by the space-filling curve makes us fill each block with points that lie close to each other and thus have a small bounding box.

Kamel and Faloutsos proposed to use the Hilbert curve~\cite{Hilbert} for this purpose. One way to describe the two-dimensional Hilbert curve is as a recursive construction that maps the unit interval $[0,1]$ to the unit square $[0,1]^2$. We subdivide the square into a grid of $2 \times 2$ square cells, and simultaneously subdivide the unit interval into four subintervals. Each subinterval is then matched to a cell; thus Hilbert's curve traverses the cells one by one in a particular order. The mapping from unit interval to unit square is refined by applying the procedure recursively to each subinterval-cell pair, so that within each cell, the curve makes a similar traversal. The traversals within these cells are rotated and/or reflected so that the traversal remains continuous from one cell to another (see Figure~\ref{fig:hilbert2d}(a)). The result is a fully-specified mapping $\curve: [0,1] \rightarrow [0,1]^2$ from the unit interval to the unit square. The mapping is easily reversed, and thanks to the fact that the curve is based on recursive subdivision of a square in quadrants, the reversed mapping can be implemented very efficiently with coordinates represented as binary numbers. This gives us a way to decide which of any two points in the unit square is the first along the curve.

We can sketch the shape of the curve by drawing, for the $k$-th level of recursion, a polygonal curve, an \emph{approximating curve} $A_k$, that connects the centres of the $4^k$ squares in the order in which they are visited. In fact, the mapping $\curve$ can also be described as the limit of the approximating curves $A_k$ as $k$ goes to infinity. Explicit descriptions of the approximating curves help us to reason about the shapes of curve sections, and thus, about the extents of their bounding boxes. For ease of notation, in this paper we scale the approximating curve for any level $k$ by a factor $2^k$ and translate it so that its vertices are exactly the points $\{0,\ldots,2^k-1\}^2$.

A $d$-dimensional version of Hilbert's curve could now be described by a series of curves $A_k$ for increasing $k$, each visiting the points $\{0,\ldots,2^k-1\}^d$, where each point corresponds to a $d$-dimensional cube of width $1/2^k$ in the unit hypercube. For $d \geq 3$, there are many ways to define such a series of curves~\cite{Alber,Haverkort3D,Harmonious}, but their distinctive properties and their differences in suitability for our purposes are largely unexplored.

\subsection{Our results}

In this paper we present a family of space-filling curves, for any number of dimensions $d \geq 3$, with two properties which we call \emph{well-foldedness} and \emph{hyperorthogonality}---Hilbert's two-dimensional curve also has these properties. We show that these properties imply that the curves have good \emph{bounding-box quality} as defined by Haverkort and Van Walderveen~\cite{Haverkort2009}.

More precisely, for any $0 \leq a \leq b \leq 1$, let $\curve([a,b])$ denote the section of the space-filling curve $\curve$ from $\curve(a)$ to $\curve(b)$, that is, $\curve([a,b])=\bigcup_{a \leq t \leq b} \{\curve(t)\}$.

The \emph{box-to-curve ratio (BCR)} of a section $\curve([a,b])$ (denoted $\BCR(\curve([a,b]))$)
is the volume of the minimum axis-aligned bounding box of $\curve([a,b])$ divided by the volume ($d$-dimensional Lebesgue measure) of $\curve([a,b])$, see Figure~\ref{fig:measures}(c).

The worst-case \BCR of a space-filling curve $f$ is the maximum \BCR over all sections of~$f$. We show that the worst-case \BCR of a well-folded, hyperorthogonal space-filling curve is at most 4, independent of the number of dimensions. Moreover, we show that this is asymptotically optimal: we prove that any $d$-dimensional space-filling curve that is described by a series of curves $A_k$ as defined above, has a section with \BCR at least $4 - O(1/2^d)$. In contrast, the $d$-dimensional ``Hilbert'' curves of Butz~\cite{Butz}, as implemented by Moore~\cite{Moore}, have sections with \BCR in $\Omega(2^{d/2})$.

In Section~\ref{sec:notation} we introduce basic nomenclature and notation. Section~\ref{sec:wellfolded} defines the concept of well-foldedness, and presents sufficient and necessary conditions for approximating curves of well-folded space-filling curves. Section~\ref{sec:hyperorthogonal} introduces the concept of hyperorthogonality. We present sufficient and necessary conditions for approximating curves of well-folded space-filling curves to be hyperorthogonal. The necessity of these conditions is then used to prove that any section of a hyperorthogonal well-folded space-filling curve has good box-to-curve ratio. Our next task is to show that hyperorthogonal well-folded curves actually exist, and this is the topic of Section~\ref{sec:existence}.

We combine the conditions from the previous sections to learn more about the shape of hyperorthogonal well-folded curves, and in particular about self-similar curves (Section~\ref{sec:selfsimilar}).
 It turns out that in two, three, and four dimensions, there are actually very few self-similar, well-folded, hyperorthogonal curves (Corollary~\ref{cor:onlytwoselfsimilar}); in five and more dimensions, more such curves exist. In Section~\ref{sec:implementation}, we make a few remarks about how to implement a comparison operator based on self-similar, well-folded, hyperorthogonal curves in any number of dimensions greater than two.
Finally, in Section~\ref{sec:discussion}, we compare the bounding box quality of hyperorthogonal well-folded curves to lower bounds and to the bounding box quality of Butz's generalization of Hilbert curves, and we discuss directions for further research.
Pseudocode for the comparison operator discussed in Section~\ref{sec:implementation} is given and explained in Appendix~\ref{sec:pseudocode}.

\subsection{Nomenclature and notation}\label{sec:notation}

\paragraph{General notation}\ \\
\textcolor{gray}{\rule[1mm]{2mm}{0.5mm}} By $D$ we denote $2^d$.\\
\textcolor{gray}{\rule[1mm]{2mm}{0.5mm}} By $\sgn(i)$ we denote the sign of $i$, that is, $\sgn(i) = -1$ if $i < 0$; $\sgn(i) = 0$ if $i = 0$, and $\sgn(i) = 1$ if $i > 0$. \\
\textcolor{gray}{\rule[1mm]{2mm}{0.5mm}} By $\flipped(i)$ we denote the function defined by $\flipped(i) = 1$ if $i < 0$, and $\flipped(i) = 0$ if $i \geq 0$. Notice $\sgn(i)=1-2\ast\flipped(i)$.

\paragraph{Vertices, edges, directions and axes}\ \\
\textcolor{gray}{\rule[1mm]{2mm}{0.5mm}} The universe in this article is the integer grid in $d$ dimensions $\mathbb{Z}^d$. \\
\textcolor{gray}{\rule[1mm]{2mm}{0.5mm}} A \emph{vertex} is a point $v = (v[1],v[2],\ldots,v[d]) \in \mathbb{Z}^d$. \\
\textcolor{gray}{\rule[1mm]{2mm}{0.5mm}} An \emph{edge} $e$ is an ordered pair of vertices $(v,w)$ with distance $||w-v|| = 1$. \\
\textcolor{gray}{\rule[1mm]{2mm}{0.5mm}} The \emph{direction} of an edge $e = (v,w)$ is the number $i \in \signeddim$ such that $w[|i|] - v[|i|] = \sgn(i)$ and $w[j] = v[j]$ if $j \neq |i|$. \\
\textcolor{gray}{\rule[1mm]{2mm}{0.5mm}} The \emph{axis} of an edge is the absolute value of its direction. Note that the edges $(v,w)$ and $(w,v)$ have opposite directions, but the same axis. \\
\textcolor{gray}{\rule[1mm]{2mm}{0.5mm}} In our figures we will use horizontal lines for axis 1, vertical lines for axis 2 and lines with another orientation for axis 3.\\
\textcolor{gray}{\rule[1mm]{2mm}{0.5mm}} By $\langle e_1,e_2,\ldots \rangle$ we denote a path of edges with directions $e_1, e_2, \ldots$.

\paragraph{Curves, length, volume, entry and exit}\ \\
\textcolor{gray}{\rule[1mm]{2mm}{0.5mm}} For the purposes of this paper, a \emph{curve} is a \emph{curve on the grid}, which is an ordered set of unique vertices where each subsequent pair of vertices forms an edge as defined above. Note that a curve never visits the same vertex more than once. Since a vertex and a direction determine an edge, a curve can alternatively be specified by the starting point and the listing of the directions of its edges in order. Note that curves are directed.\\
\textcolor{gray}{\rule[1mm]{2mm}{0.5mm}} A \emph{space-filling curve} is always a mapping $f: [0,1] \rightarrow [0,1]^d$, while any other curve discussed in this paper will be assumed to be a curve on the grid.\\
\textcolor{gray}{\rule[1mm]{2mm}{0.5mm}} A \emph{free} curve is a curve without a starting point, so with unspecified location: it is described by the directions of its edges only.\\
\textcolor{gray}{\rule[1mm]{2mm}{0.5mm}} The \emph{reverse} $\reverse{C}$ of a free curve $C$ is obtained by reversing the order of the edge directions \emph{and} reversing the directions themselves, which means negating them. \\
\textcolor{gray}{\rule[1mm]{2mm}{0.5mm}} The \emph{length} of a curve is the number of edges, the \emph{volume} of a subset of the grid is the number of vertices it contains. So the volume $\volume{C}$ of a curve $C$ is its length + 1. \\
\textcolor{gray}{\rule[1mm]{2mm}{0.5mm}} The first vertex of a curve is called the \emph{entry}; the last vertex is called the \emph{exit}.

\paragraph{$k$-Curves and $k$-cubes}\ \\
\textcolor{gray}{\rule[1mm]{2mm}{0.5mm}} A \emph{$k$-cube} is a $d$-dimensional cube with $\twodk $ points, so with a side of length $2^k -1$.\\
\textcolor{gray}{\rule[1mm]{2mm}{0.5mm}} A \emph{$k$-curve} is a Hamiltonian path on the integer grid in a \emph{$k$-cube}. \\
\textcolor{gray}{\rule[1mm]{2mm}{0.5mm}} Since each of the (integer) points of the cube is visited by the curve exactly once, its volume is $\twodk$  and the length of a $k$-curve is $\twodk - 1$.

\paragraph{Approximating curves}\ \\
\textcolor{gray}{\rule[1mm]{2mm}{0.5mm}} The space-filling curves under study in this paper will be approximated by curves on the grid as just defined. By $A_0, A_1, \ldots$ we will denote a sequence of curves that approximates a $d$-dimensional space-filling curve, where $A_0$ is a single vertex and $A_k$ is a $k$-curve. \\
\textcolor{gray}{\rule[1mm]{2mm}{0.5mm}} By $v_{k,1},v_{k,2},\ldots,v_{k,K}$, where $K = \twodk$, we denote the vertices of $A_k$ in order, and by $e_{k,i}$ we denote the direction of the edge $(v_{k,i},v_{k,i+1})$. \\ \textcolor{gray}{\rule[1mm]{2mm}{0.5mm}} Each vertex $v_{k,i}$ of $A_k$ represents a $d$-dimensional hypercube $H_{k,i}$ of width $1/2^k$ that is visited by the space-filling curve approximated by $A_k$. The vertices $v_{k+1,D\ast i - D + 1},\ldots,v_{k+1,D\ast i}$ of $A_{k+1}$ model the order in which the space-filling curve traverses the $d$-dimensional hypercubes of width $1/2^{k+1}$ whose union is $H_{k,i}$.

\begin{figure}[t]
\begin{center}
\includegraphics[scale=1.3]{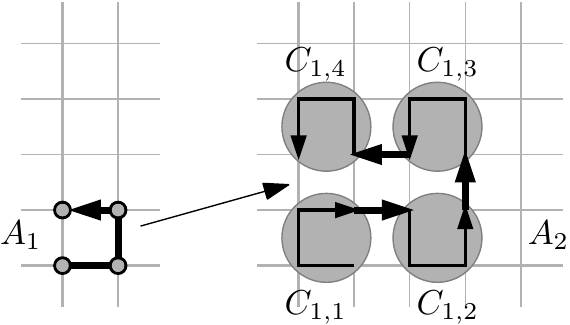}
\hfill
\raise 1.5\baselineskip\hbox{\includegraphics{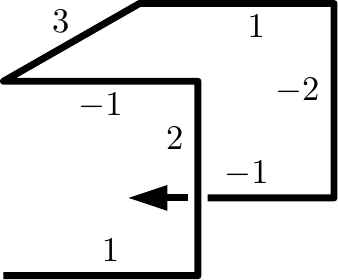}}
\caption{\emph{Left:} A parent curve $A_1$ is inflated to create $A_2$, which is composed of the child curves $C_{1,1}$, $C_{1,2}$, $C_{1,3}$ and $C_{1,4}$, and edges of $A_1$ which are translated such that they connect the child curves to each other at their end points.\quad \emph{Right:} $G(3)$ with the directions of its edges.\label{fig:inflating}\label{fig:G3}}
\end{center}
\end{figure}

\noindent \textcolor{gray}{\rule[1mm]{2mm}{0.5mm}} Therefore it must be possible to construct $A_{k+1}$ from $A_k$, which we call the \emph{parent curve}, by \emph{inflation}: we replace each vertex $v_{k,i}$ of the parent curve with a 1-curve $C_{k,i}$ (a \emph{child curve}), whose vertices are those of the unit cube, translated by $2 \ast v_{k,i}$.
Each edge $(v_{k,i},v_{k,i+1})$ of the parent curve is replaced by an edge $(v_{k+1,D\ast i},v_{k+1,D\ast i+1})$ in $A_{k+1}$ of the same direction, connecting the exit of $C_{k,i}$ to the entry of $C_{k,i+1}$, see Figure~\ref{fig:inflating}, left. Note that not just any choice of child curves results in a valid $(k+1)$-curve. The 1-curves that replace the vertices have to be chosen carefully such that for each edge $(v_{k,i},v_{k,i+1})$ of the parent curve, there is indeed an edge in the grid from the exit of $C_{k,i}$ to the entry of $C_{k,i+1}$. In Section~\ref{sec:continuityconditions} we will discuss how the $1$-curves should be constructed so that they match up.

In what follows, the first subscripts to $v$, $e$, $H$ and $C$ will usually be omitted if they are clear from the context, for example, if $k$ is fixed, or if the approximating curve in question is otherwise specified. For example, ``a child curve $C_i$ of $A_{k-1}$'' should be read as: ``a child curve $C_{k-1,i}$'' (and it would be a subcurve of $A_k$).

Observe that our definition of curves on the grid restricts the generalizations of Hilbert curves under study to \emph{face-continuous curves}, that is, each pair of consecutive $d$-dimensional hypercubes along the space-filling curve must share a $(d-1)$-dimensional face. In Section~\ref{sec:lowerbounds}, we will discuss why, in the context of this paper, this restriction is justified.

\section{Well-folded curves}\label{sec:wellfolded}

\subsection{Gray codes and definition of well-folded curves}

In the process of inflating, we will restrict ourselves in this paper to replacing vertices with isometric images (like translations, rotations,  reflections or taking the reverse, shortly: all  distance-preserving mappings) of one particular 1-curve, namely the free curve $G(d)$ that follows the so-called \emph{binary reflected Gray code}.

\begin{definition}\label{def:graycode}
The free curve $G(d)$ is defined recursively as follows: $G(0)$ is empty; $G(d)$ is the concatenation of $G(d-1)$, $\langle d\rangle$, and $\rule{0mm}{5mm}\reverse{G(d-1)}$.
\end{definition}


For example, $G(2)$ is the free curve $\langle 1,2,-1\rangle$ (Figure \ref{fig:inflating}, left), $G(3)$ is shown in Figure~\ref{fig:G3}, right, and $G(4)$ is the free curve $\langle 1,2,-1,3,1,-2,$ $-1,4,1,2,-1,-3,1,-2,-1\rangle$.

The length of $G(d)$ is, by induction, $2^d -1$, which is the maximum length of a Hamiltonian path on the unit cube in $\mathbb{Z}^d$.

The following properties of $G(d)$ are well-known:

\begin{lemma}\label{lem:graycodealternatingpattern}
If $d\geq 2$, then, in $G(d)$ as well as $\reverse{G(d)}$, edges with axis 1 and edges with other axes alternate, starting with $\langle 1\rangle$ and ending with $\langle-1\rangle$.
\end{lemma}
\begin{proof}
Straightforward by induction on increasing $d$, with base case $d=2$.
\end{proof}


\begin{lemma}\label{lem:graycodeendpoints}
Let $G^0(d)$ be $G(d)$ with entry point $(0,\ldots,0)$, where $d \geq 1$. Then the vertices of $G^0(d)$ are those of the 1-cube $\{0,1\}^d$ and the exit point is the point $v = (0,\ldots,0,1)$ with $v[d] = 1$ and $v[j] = 0$ for $j < d$.
\end{lemma}
\begin{proof}
We can prove this by induction on increasing $d$. The base case $d = 1$ is easy to verify. Now suppose the lemma holds for $d-1$, that is, the vertices of $G^0(d-1)$ are those of the 1-cube $\{0,1\}^{d-1}$ and the exit point is $(0,\ldots,0,1)$. Or to put it differently, the exit point of $G^0(d-1)$ is the endpoint of an edge $\langle d-1\rangle$ that starts in the origin.

Now recall that $G(d)$ is the concatenation of $G(d-1)$, $\langle d\rangle$ and $\reverse{G(d-1)}$. Therefore the exit point of $G^0(d)$ can be found as the endpoint of a curve $\langle(d-1),d,-(d-1)\rangle$, starting in the origin. Clearly, this endpoint is $(0,\ldots,0,1)$. The vertices of $G^0(d)$ are those of $G(d-1)$ starting in the origin plus those of $\reverse{G(d-1)}$ ending at $(0,\ldots,0,1)$, or equivalently, those of $G(d-1)$ starting in the origin plus those of $G(d-1)$ starting at $(0,\ldots,0,1)$. Together these constitute the set $\{0,1\}^d$.
\end{proof}

The following lemma will prove useful in Sections \ref{sec:hyperorthogonal} and \ref{sec:comparison} but can be skipped on first reading:

\begin{lemma}\label{lem:prefixspan}
The axes of the first (and last) $n$ edges of $G(d)$ constitute the set $\{1,\ldots,m\}$, where $m = 1 + \lfloor \log_2(n) \rfloor  = \lceil \log_2(n + 1) \rceil$.
\end{lemma}
\begin{proof}
For $1 + \lfloor \log_2(n) \rfloor = m = \lceil \log_2(n+1)\rceil$ we have
$2n \geq 2^m \geq n+1$ and thus, $n \geq 2^{m-1}$ and $n \leq 2^m-1$. Therefore the first $n$ edges of $G(d)$ include at least a full $G(m-1)$ and an edge $\langle m\rangle$, and not more than a full $G(m)$. In a symmetric way, the last $n$ edges of $G(d)$ include at least an edge $\langle -m\rangle$ and a full $\reverse{G(m-1)}$, and not more than a full $\reverse{G(m)}$. It follows that the first or last $n$ edges cover $m$ different axes.
\end{proof}

\begin{definition}\label{def:wellfolded}
A curve is \emph{well-folded} if it is a single vertex, or if it is obtained by inflating a well-folded curve by replacing its vertices by isometric images of $G(d)$. A space-filling curve is \emph{well-folded} if its approximating curves are well-folded.
\end{definition}

Note that in two dimensions, all possible 1-curves are in fact isometric images of $G(2)$, so any face-continuous space-filling curve based on recursive subdivision of a square into four squares must be well-folded (for example, Hilbert's curve or the $\beta\Omega$-curve~\cite{Wierum}).

In higher dimensions, the most common generalizations of the Hilbert curve are well-folded as well, but there are also face-continuous curves based on recursive subdivision of a cube into eight cubes that are not well-folded (using generators of types B and C from Alber and Niedermeier~\cite{Alber,Haverkort3D}). In Section~\ref{sec:discussion}, we will briefly get back to non-well-folded curves; until then, we will focus on well-folded curves.

\subsection{Notation for isometries of Gray codes in well-folded curves}
\label{sec:signedpermutations}

The isometric transformations of 1-curves which we need in this paper are those of the hyperoctahedral group of symmetries of the hypercube. This group is the product of the symmetric group ${{S}_{d}}$ (the group of all permutations of the $d$ coordinate axes) and the group of $2^d$ reflections formed by all combinations of reflections in hyperspaces orthogonal to the coordinate axes. Thus there are $d! \ast 2^d$ such transformations.

To distinguish these transformations, we will use \emph{signed permutations}. A signed permutation $\pi$ is a bijection from $\signeddim$ to itself with the property that $\pi(-k) = -\pi(k)$ for $k \in \unsigneddim$. It is denoted by $[\pi(1),\pi(2),\ldots,\pi(d)]$.

Given a $k$-cube $H$, a signed permutation $\pi$ specifies the isometry that maps $H$ onto itself and maps the direction $k$ to the direction $\pi(k)$. If $\pi$ is a signed permutation, then $\pi(\mathcal{X})$ denotes the application of $\pi$ to all elements of the vector, set, or sequence $\mathcal{X}$; $|\pi|$ denotes the permutation $[|\pi(1)|,|\pi(2)|,\ldots,|\pi(d)|]$; and $\pi^{-1}$ denotes the inverse of $\pi$, that is, $\pi^{-1}(x) = y$ if and only if $\pi(y) = x$.


Note that signed permutations do not allow us to express the isometric transformation that consists of reversing a curve. For now, this is not a problem, because the reversal of $G(d)$ is identical to its reflection in coordinate $d$.

We define the \emph{orientation} of an isometry of $G(d)$ as the direction of the vector from entry to exit. 
A direct corollary of Lemma~\ref{lem:graycodeendpoints} is the following:

\begin{corollary}\label{cor:transform1curve}
Let $G^0(d)$ be $G(d)$ with entry point $(0,\ldots,0)$, and let $\pi$ be a signed permutation. Then the coordinates of the entry point $a$ of $\pi(G(d))$ are given by $a[j] = \flipped\big(\pi^{-1}(j)\big)$ for $j \in \unsigneddim$; the orientation of $\pi(G(d))$ is $\pi(d)$; and the coordinates of the exit point $b$ of $\pi(G(d))$ satisfy $b[j] = 1 - a[j]$ for $j = |\pi(d)|$ and $b[j] = a[j]$ for $j \neq |\pi(d)|$.
\end{corollary}

\paragraph{Isometries in approximating curves}
Consider a sequence of well-folded approximating curves $A_0, A_1, \ldots$. By $\perm_{k,i}$ we denote the transformation (modulo translation) that is applied to $G(d)$ to obtain the 1-curve $C_{k,i}$ that replaces vertex $v_i$ of $A_k$ in the inflation of $A_k$ to $A_{k+1}$. For example, for the curves in Figure~\ref{fig:inflating}, left, we have $\perm_{1,1} = [-1,2]$; $\perm_{1,2} = [-2,1]$; $\perm_{1,3} = \perm_{1,4} = [2,-1]$. As with $v$, $e$, $H$ and $C$, the first subscript will usually be omitted if it is clear from the context.

\subsection{Conditions on edges and isometries in well-folded curves}
\label{sec:continuityconditions}

As observed before, when inflating a curve $A_k$, the 1-curves that replace the vertices of $A_k$ have to be chosen carefully such that the exit of $C_i$ and the entry of $C_{i+1}$ constitute an edge with direction $e_i$. For this we need the conditions as stated in Theorem~\ref{thm:wellfolded} below.

\begin{theorem}\label{thm:wellfolded}
Given a well-folded approximating curve $A_k$ for a particular, fixed level $k\ge 0$.
Inflating $A_k$ to $A_{k+1}$ results in a well-folded approximating curve $A_{k+1}$ if and only if, for each $1 \leq i < \twodk$:
\begin{itemize}
\item for $j \in \{1,\ldots,d\}$ we have $\sgn\big(\perm^{-1}_{i+1}(j)\big) = \sgn\big(\perm^{-1}_i(j)\big)$ if and only if $j$ equals neither or both of $|\perm_i(d)|$ and $|e_i|$; otherwise $\sgn\big(\perm^{-1}_{i+1}(j)\big) = -\sgn\big(\perm^{-1}_i(j)\big)$;
\item $\sgn\big(\perm^{-1}_{i+1}(e_i)\big) = 1$.
\end{itemize}
\end{theorem}

\begin{proof}
By construction, $A_{k+1}$ is obtained by replacing the vertices of $A_k$ by isometric images of $G(d)$. The challenge is to prove that the above conditions are necessary and sufficient to guarantee that $A_{k+1}$ is indeed a curve, and hence, well-folded.

Recall that $C_i=\perm_i(G(d))$ is the 1-curve (modulo translation) that replaces vertex $v_i$ of $A_k$ in the process of inflation and $e_i=(v_i,v_{i+1})$. For a given $i$, let $a$, $b$ and $c$ be, respectively, the \emph{entry} of $C_i$, the \emph{exit} of $C_i$, and the \emph{entry} of $C_{i+1}$, all relative to the point $2\ast v_i$. By Corollary~\ref{cor:transform1curve} we have $a[j] = \flipped\big(\perm^{-1}_i(j)\big)$ and $c[j] = \flipped\big(\perm^{-1}_{i+1}(j)\big) \pmod 2$ for $j \in \unsigneddim$.
Our task is to establish the conditions under which $(b,c)$ is indeed an edge with direction $e_i$, that is: for $j = |e_i|$ we should have $c[j] = b[j] + \sgn(e_i)$, and for $j \neq |e_i|$ we should have $c[j] = b[j]$.

Note that if $(b,c)$ is an edge with direction $e_i$, then the path $\langle \perm_i(d), e_i\rangle$ brings us from $a$ via $b$ to $c$. Since each edge increments or decrements one coordinate by 1, it follows that if and only if $j$ equals neither of both of $|\perm_i(d)|$ and $|e_i|$, we have $c[j] = a[j] \pmod 2$ and thus, $\flipped\big(\perm^{-1}_i(j)\big) = \flipped\big(\perm^{-1}_{i+1}(j)\big)$. This proves that the first condition of the theorem is necessary.

Conversely, the first condition of the theorem, together with Corollary~\ref{cor:transform1curve}, gives us that $b[j] = c[j] \pmod 2$ if and only if $j \neq |e_i|$, as witnessed by the following table. The third column is equivalent with the first condition of the theorem, as derived in the previous paragraph. The fourth column follows from Corollary~\ref{cor:transform1curve}. The fifth column follows from the third and the fourth.

\noindent\hbox to\hsize\bgroup\hfill
\begin{tabular}{cc|ccc}
$j = |\perm_i(d)|$ & $j = |e_i|$ & $a[j] = c[j] \pmod 2$ & $a[j] = b[j] \pmod 2$ & $b[j] = c[j] \pmod 2$ \\\hline
 no &  no & yes & yes & yes \\
 no & yes &  no & yes &  no \\
yes &  no &  no &  no & yes \\
yes & yes & yes &  no &  no \\
\end{tabular}\hfill\egroup

\noindent In fact, since $v_i[j]$ and $v_{i+1}[j]$ are equal if $j \neq |e_i|$, we get $b[j] = c[j]$ (without $\pmod 2$) if $j \neq |e_i|$. On the other hand, if $j = |e_i|$, we have, so far, only established $b[j] \neq c[j] \pmod 2$.

To complete the proof of Theorem~\ref{thm:wellfolded}, we will now show that, given $j = |e_i|$ and $c[j] \neq b[j] \pmod 2$, we actually have $c[j] = b[j] + \sgn(e_i)$ if and only if the second condition of the theorem is satisfied. In fact, the second condition, $\sgn\big(\perm^{-1}_{i+1}(e_i)\big) = 1$, expresses that, within the 1-cube filled by $C_{i+1}$, the entry is on the side that is adjacent in direction $e_i$ to the 1-cube filled by $C_i$. To analyse this in more detail, we distinguish two cases: first, $\sgn(e_i) = 1$, and second, $\sgn(e_i) = -1$.

If $\sgn(e_i) = 1$, then $v_{i+1}[j] = v_i[j] + 1$. Hence, by the fact that $A_{k+1}$ is obtained by inflation from $A_k$, we have $0 \leq b[j] \leq 1$ and $2 \leq c[j] \leq 3$. Moreover, given $c[j] \neq b[j] \pmod 2$, we have that $(b[j], c[j])$ is either $(0,3)$ or $(1,2)$. So we have $c[j] = b[j] + \sgn(e_i) = b[j] + 1$ if and only if $c[j] = 2 = 0 \pmod 2$. By Corollary~\ref{cor:transform1curve} this is the case if and only if $\flipped\big(\perm^{-1}_{i+1}(j)\big) = 0$; with $j = |e_i| = e_i$ we can rewrite this as $\sgn\big(\perm^{-1}_{i+1}(e_i)\big) = 1$.

Similarly, if $\sgn(e_i) = -1$, then $v_{i+1}[j] = v_i[j] - 1$. Hence, by the fact that $A_{k+1}$ is obtained by inflation from $A_k$, we have $0 \leq b[j] \leq 1$ and $-2 \leq c[j] \leq -1$. Moreover, given $c[j] \neq b[j] \pmod 2$, we have that $(b[j], c[j])$ is either $(0,-1)$ or $(1,-2)$. So we have $c[j] = b[j] + \sgn(e_i) = b[j] - 1$ if and only if $c[j] = -1 = 1 \pmod 2$. By Corollary~\ref{cor:transform1curve} this is the case if and only if $\flipped\big(\perm^{-1}_{i+1}(j)\big) = 1$; with $j = |e_i| = -e_i$ we can rewrite this as $\sgn\big(\perm^{-1}_{i+1}(e_i)\big) = \sgn\big(\perm^{-1}_{i+1}(-j)\big) = -\sgn\big(\perm^{-1}_{i+1}(j)\big) = 1$.
\end{proof}

Given the edges of $A_k$ and the signs of the inverse permutations, Theorem~\ref{thm:wellfolded} allows us to determine the last elements $\sigma_i(d)$ of each permutation.
Conversely, given the edges and the last elements of each permutation, Theorem~\ref{thm:wellfolded} allows us to determine the signs of each permutation.
Note that this leaves $d-1$ elements of each $|\perm_i|$ unspecified and without consequence: any permutation of those elements will do.

\begin{observation}\label{obs:startingpoint}
Let $f$ be a well-folded space-filling curve approximated by $A_0,A_1,\ldots$, and let $x = f(0)$ be the starting point of $f$. 
Then $x[j] = \sum_{k=0}^\infty \flipped(\perm^{-1}_{k,1}(j)) / 2^{k+1}$.\\
In other words, the digits of the binary representation of $x[j]$ behind the fractional point are\\ $\flipped(\perm^{-1}_{0,1}(j)),\flipped(\perm^{-1}_{1,1}(j)),\flipped(\perm^{-1}_{2,1}(j)),\ldots$.
\end{observation}

\section{Hyperorthogonal well-folded curves}\label{sec:hyperorthogonal}

So far, we have been defining and discussing properties of curves that are in fact common to the previously best-known generalizations of Hilbert's curve to higher dimensions. We will now introduce a new property that is \emph{not} satisfied by any of the previously known generalizations that we are aware of, and which will prove useful in designing novel curves with good box-to-curve ratios.

\subsection{Definition and characterization}

\begin{definition}\label{def:hyperorthogonal}
We call a curve \emph{hyperorthogonal} if and only if, for any $n \in \{0,\ldots,d-2\}$,
each sequence of $2^{n}$ consecutive edges have exactly $n+1$ different axes.
A space-filling curve is hyperorthogonal if its approximating curves are hyperorthogonal.
\end{definition}

Notice that an $n$-dimensional 1-cube (in $\mathbb{Z}^d$) can hold at most ${2^n}-1$ consecutive edges of a curve, so any curve constructed by inflation contains sets of ${2^n}$ edges that have at least $n+1$ different axes, for each $n \leq d-1$. Hyperorthogonality requires that this holds for \emph{every} set of $2^n$ edges, provided $n \leq d-2$. The definition leaves little room for being made
more strict (see Inset~\ref{ins:strict}).

\begin{inset}
\caption{No room for a stricter definition of hyperorthogonality}
\label{ins:strict}
In Definition~\ref{def:hyperorthogonal}, the upper bound on $n$ cannot be raised to $d-1$, as this would require that, in two dimensions, no pair of consecutive edges would have the same direction. It is easy to see that a 2-curve with this property cannot be constructed.

\leavevmode\hspace{2em}
Consider a square of four by four vertices, that is, a 2-cube for $d = 2$. Within this square, four vertices lie in a corner. Of these four corner vertices, let ${s}$ be the second one visited by the curve. The curve must visit at least two vertices before ${s}$ and at least two vertices after ${s}$ in order to reach the other corners of the square. Let $S$ be the unit square (quadrant of four vertices) that contains ${s}$.

\leavevmode\hspace{2em}
Now consider the sequence that consists of the two edges that precede ${s}$ and the two edges that follow ${s}$. Since this sequence visits five vertices, it clearly does not fit in $S$, and therefore the $2^{d-1} = 2$ edges preceding ${s}$ or the $2^{d-1} = 2$ edges following ${s}$ do not fit in a two-dimensional unit cube. Hence, either two edges preceding $s$ or the two edges following $s$ must be collinear.

\medskip
Hyperorthogonality still allows that less than $2^{n}$, but more than $2^{n-1}$, consecutive edges also span a $(n+1)$-dimensional space---this is also necessary, since otherwise, even if $d \geq 3$, any three consecutive edges would be restricted to alternating between two dimensions, which, by induction, would restrict the whole curve to edges alternating between two dimensions.
\end{inset}

For $d=2$, hyperorthogonality requires only that each single edge spans a one-dimensional space, which is obvious. So all two-dimensional curves are hyperorthogonal.

For $d=3$ each two consecutive edges must span a two-dimensional space, so each pair of consecutive edges must be orthogonal. (For that reason the property is called `hyperorthogonal' for higher dimensions as well.)

Note that $G(d)$ is hyperorthogonal for all $d$.

As can be seen by inspecting familiar generalizations of Hilbert curves to three dimensions, if we construct a sequence of curves $A_0,\ldots,A_k$ in three or more dimensions by inflation, using isometric images of $G(d)$ to inflate vertices, then $A_k$ is not necessarily hyperorthogonal, even though $G(d)$ is (see, for example, the Butz-Moore curve in Figure~\ref{fig:3dcurves}, right, where there are two collinear edges along the top back edge of the cube). The next theorem states what conditions the isometries should fulfill in order to obtain hyperorthogonal curves.

\begin{definition}\label{def:depth}
The \emph{depth} of a direction $a$ in a signed permutation $\pi$, denoted $\permdp(\pi,a)$, is defined as follows: if $|a| \in \{|\pi(d)|, |\pi(d-1)|\}$, then $\permdp(\pi,a) = 0$, otherwise $\permdp(\pi,a)$ is the number $j$ such that $|\pi(d-1-j)| = |a|$.
\end{definition}
\noindent So the depth of $\pi(d-2)$ is $1$, the depth of $\pi(1)$ is $d-2$, and since each axis occurs in $\pi$, each direction has a depth in $\pi$.

\begin{theorem}\label{thm:hyperorthogonal}
For fixed $k$, let $K=\twodk $, and let $A_0,\ldots,A_{k+1}$ be a sequence of well-folded curves constructed by inflation (with all the associated notation introduced in the previous sections). Suppose $A_k$ is hyperorthogonal, then $A_{k+1}$ is hyperorthogonal as well if and only if the following conditions are satisfied:

\begin{enumerate}
\item\label{item:1} for each $i \in \{1,\ldots,K-1\}$: $\permdp(\perm_{k,i}, e_{k,i}) = 0= \permdp(\perm_{k,i+1}, e_{k,i}) $;
\item\label{item:2} for each $i \in \{1,\ldots,K-1\}$ and each direction $a$:\\$|\permdp(\perm_{k,i}, a) - \permdp(\perm_{k,i+1}, a)| \leq 1$.
\end{enumerate}
\end{theorem}

\begin{proof}
As usual, we will omit the subscripts $k$ in this proof.

\textit{Necessity:}
Suppose condition \ref{item:1} is violated, that is, $|e_i| \in \{|\perm_i(1)|,\ldots,|\perm_i(d-2)|\}$ or\\ $|e_i| \in \{|\perm_{i+1}(1)|,\ldots,|\perm_{i+1}(d-2)\}$.
We analyse the first case $|e_i| \in \{|\perm_i(1)|,\ldots,|\perm_i(d-2)|\}$, the second case is symmetric.\\
Consider the last $2^{d-2} - 1$ edges of $\perm_i(G(d))$: by Lemma~\ref{lem:prefixspan}, their axes form the set $\{|\perm_i(1)|,\ldots,|\perm_i(d-2)|\}$. In $A_{k+1}$, these edges will be followed by an edge with axis $|e_i| \in \{|\perm_i(1)|,\ldots,|\perm_i(d-2)|\}$. Thus we get a sequence of $2^{d-2}$ edges with only $d-2$ different axes: too few for hyperorthogonality as defined by Definition~\ref{def:hyperorthogonal}.

Now suppose condition \ref{item:1}  is satisfied, but condition \ref{item:2} is violated, that is, there are $h$ and $j$ such that $|\perm_i(h)| = |\perm_{i+1}(j)|$ or $|\perm_i(j)| = |\perm_{i+1}(h)|$ and $h+1 < j \leq d-1$ or $h+2 < j = d$.
We analyse the first case $|\perm_i(h)| = |\perm_{i+1}(j)|$, the second case is symmetric. \\
Observe that $|\perm_i(1)|,\ldots,|\perm_i(h)|$ all differ from $|e_i|$, since $h \leq d-3$. Also,  in $\perm_{i+1}(G(d))$, the axes $|\perm_{i+1}(1)|,\ldots,|\perm_{i+1}(h+1)|$ all differ from $|e_i|$ as well as from $|\perm_{i+1}(j)| = |\perm_i(h)|$. Now consider the sequence of $2^{h+1}$ edges that consists of the last $2^{h}-1$ edges of $\perm_i(G(d))$, followed by the edge with direction $e_i$, and the first $2^{h}$ edges of $\perm_{i+1}(G(d))$. By Lemma~\ref{lem:prefixspan}, the last $2^h$ edges of this sequence have $h+1$ different axes $|\perm_{i+1}(1)|,\ldots,|\perm_{i+1}(h+1)|$, while the first $2^h$ edges contribute two more different axes, namely $|\perm_i(h)|$ and $|e_i|$. Thus there are $h+3$ different axes in this sequence of $2^{h+1}$ edges: too many for hyperorthogonality as defined by Definition~\ref{def:hyperorthogonal}.

\textit{Sufficiency:}
We distinguish two cases: a sequence of $2^n$ edges $E$, with $n \leq d-2$, either lies within a single 1-curve $\perm_i(G(d))$, or not.

In the first case, let $j$ be the highest index such that $E$ includes an edge with axis $\perm_i(j)$. From Definition~\ref{def:graycode} we get that the axes of the edges of $E$ are, in order, for some $m \leq 2^n-1$, those of the last $m$ edges of $\perm_i(G(j-1))$, followed by $\perm_i(j)$ and the axes of the first $2^n - m - 1$ edges of $\perm_i(\reverse{G(j-1)})$.
 Since either $m$ or $2^n - m - 1$ must be at least $2^{n-1}$, it follows from Lemma~\ref{lem:prefixspan} that the axes of these edges are exactly $\{|\perm_i(1)|,\ldots,|\perm_i(n)|,|\perm_i(j)|\}$, since $j > n$. Hence, the edges of $E$ have exactly $n+1$ different axes and satisfy the conditions for hyperorthogonality.

In the second case, $E$ consists of the last $m$ edges in a 1-curve $\perm_i(G(d))$, followed by $\langle e_i\rangle$, and $2^n - m - 1$ edges in $\perm_{i+1}(G(d))$. Assume $m > 2^n - m - 1$ (the opposite case is symmetric), and hence, $2^{n-1} \leq m \leq 2^n - 1$ and $2^n - m - 1 \leq 2^{n-1} - 1$. By Lemma~\ref{lem:prefixspan}, the first $m$ edges have axes $\{|\perm_i(1)|,\ldots,|\perm_i(n)|\}$. Since $|e_i| \in \{|\perm_i(d-1)|,|\perm_i(d)|\}$ and $d-1 > n$, the edge $\langle e_i\rangle$ contributes one more axis. The remaining edges have axes from $\{|\perm_{i+1}(1)|,\ldots,|\perm_{i+1}(n-1)|\}$, which, because of the second condition of the theorem, is a subset of $\{|\perm_i(1)|,\ldots,|\perm_i(n)|\}$; hence these edges do not contribute any more axes. In total, the edges of $E$ have exactly $n+1$ different axes and satisfy the conditions for hyperorthogonality.
\end{proof}

\subsection{Box-to-curve ratio $\le 4$}

To bound the box-to-curve ratio (\BCR) of sections of hyperorthogonal well-folded space-filling curves, we will make use of the following lemma:

\begin{lemma}\label{lem:fitscubes}
For any $n \in \{0,1,\ldots,d-2\}$, each sequence of $2^n$ consecutive edges of a well-folded, hyperorthogonal curve lies inside an $(n+1)$-dimensional unit cube.
\end{lemma}
\begin{proof}
Definition~\ref{def:hyperorthogonal} states that each sequence $E$ of $2^n$ consecutive edges of a hyperorthogonal curve lies inside an axis-aligned box that has non-zero width in exactly $n+1$ dimensions. Therefore, to prove the lemma, we  only have to show that in none of these $n+1$ dimensions, the width is more than~1.

We can prove this by inspection of the sufficiency proof of Theorem~\ref{thm:hyperorthogonal}.
In the first case, the width is not more than 1 in any dimension, since all of $E$ lies inside a single unit cube.
In the second case, $E$ lies in the union of two unit cubes, which is a box with width 3 in dimension $|e_i|$, and width 1 in the remaining dimensions. However, as the proof argues, $E$ contains only one edge with axis $|e_i|$; hence the width in this dimension is only one.
\end{proof}

\begin{theorem}\label{thm:bcr}
The box-to-curve ratio of any section of a hyperorthogonal well-folded space-filling curve is at most~4.%
\end{theorem}
\begin{proof}
Consider a section $s$ of a hyperorthogonal well-folded space-filling curve $f$, approximated by a series of curves $A_0,A_1,\ldots$.
Let $E_k$ be the subcurve of $A_k$ that contains all vertices $v_i$, representing hypercubes $H_i$ of width $1/2^k$, whose interiors are intersected by $s$.
For $\{v_h,\ldots,v_j\}=E_k$, the bounding box of $s$ is contained in the smallest axis-aligned box that fully contains all hypercubes $H_h,\ldots,H_j$.

Now let $k$ be the smallest index such that $E_{k+1}$ contains at least one vertex that represents a hypercube of width $1/2^{k+1}$ that is fully contained in $s$. By this choice of $k$, the subcurve $E_k$ of $A_k$ does not contain any vertex that represents a hypercube of width $1/2^k$ that is fully contained in $s$. Thus, $E_k$ contains only a single vertex $x$, or two vertices $x$ and $y$, and $E_{k+1}$ consists of vertices from the respective child curves $C_x$ and $C_y$ that replace $x$ and $y$ in the inflation from $A_k$ to $A_{k+1}$.

Note that this implies that the bounding box of $E_{k+1}$ has at most the volume of two 1-cubes, that is $2^{d+1}$. Define $E = E_{k+1}$, let $X$ be the maximum common subcurve of $C_x$ and $E$, and, if $y$ exists, let $Y$ be the maximum common subcurve of $C_y$ and $E$, otherwise $Y = \emptyset$. Thus, $\volume{Y} = \volume{E} - \volume{X}$; without loss of generality, assume $\volume{Y} \leq \volume{X}$. Furthermore, let $c = |e_{\min(x,y)}|$ be the axis of the connecting edge of $X$ and $Y$.

A number of cases with smartly chosen boundaries for $\volume{E}$, $\volume{X}$ and $\volume{Y}$ can now be distinguished, as shown in the table below. In each case, we derive an upper bound $\mathit{MaxBoxVol}$ on the bounding box volume, and a lower bound $\mathit{MinCrvVol}$ on the number of vertices of $E$ that represent hypercubes completely covered by $s$ (this is usually all of $E$ except for the first and last vertex). From this we can derive that the box-to-curve ratio is less than $\mathit{MaxBoxVol} / \mathit{MinCrvVol} \leq 4$.

\addvspace\baselineskip
\noindent\hbox to\hsize{\hfill\def\arraystretch{1.4}\begin{tabular}{|c|l|c|c|}\hline
\multicolumn{2}{|@{\,}l@{\,}|}{\rule[-1mm]{0mm}{5.5mm} Case}& $\mathit{MaxBoxVol}$ & $\mathit{MinCrvVol}$ \\
\hline\hline
\textbf{A} & $2^{d-1} +2 \leq \volume{E} \leq 2^{d+1}$ & $2^{d+1}$    & $2^{d-1}$ \\
\hline
B & $2^{d-2} +2 \leq \volume{E} \leq 2^{d-1} +1$ and... & & \\
B1& ...and $\volume{Y} \leq \volume{X} \leq 2^{d-2}$
& $2^d $        & $2^{d-2} $ \\
B2& ...and $2^{d-3} < \volume{Y} \leq 2^{d-2} < \volume{X}$
& $\frac32 \cdot\, 2^d $ & $\frac32 \cdot 2^{d-2} $ \\
B3& ...and $1 \leq \volume{Y} \leq 2^{d-3}$
& $2^d $        & $2^{d-2} $ \\
\textbf{B4}& ...and $\volume{Y} = 0$
& $2^d $        & $2^{d-2} $ \\
\hline
C& $3 \leq \volume{E} \leq 2^{d-2} + 1$
& $4 (\volume{E}-2) $ & $\volume{E} - 2$\\
\hline
\textbf{D}& $\volume{E} \leq 2$
& $2 $          & 1 \\
\hline
\end{tabular}\hfill}
\par\addvspace\baselineskip

\noindent Note that B1, B2, B3, and B4 are subcases for the same bounds on $\volume{E}$, where B1 is the case of having small $X$, and B2, B3, and B4 are the cases of large $X$ with various bounds on the size of $Y$. For cases  \textbf{A}, \textbf{B4}, and  \textbf{D} the bounds on the bounding box volume are trivial; cases B1, B2, B3, and C require a more careful analysis.

\begin{enumerate}
\item[B1:] By Theorem~\ref{thm:hyperorthogonal}, for the axis $c$ of the connecting edge between $X$ and $Y$ we have $\permdp(\perm_x,c) = 0$. Since $\volume{X} \leq 2^{d-2}$, and thus, the length of $X$ is at most $2^{d-2}-1$, Lemma~\ref{lem:prefixspan} now tells us that the edges of $X$ have axes from $|\perm_x(1)|,\ldots,|\perm_x(d-2)|$, hence not including $c$. Therefore $X$ is included in the half-cube ($(d-1)$-dimensional 1-cube) that consists of the vertices of $C_x$ that are adjacent to vertices of $C_y$. Likewise, $Y$ is included in the half-cube that consists of the vertices of $C_y$ that are adjacent to $C_x$. These two half-cubes together constitute a $d$-dimensional unit cube of volume $2^d$.

\item[B2:] As in case B1, $Y$ is included in the half-cube that consists of the vertices of $C_y$ that are adjacent to $C_x$. This half-cube, together with $C_x$, has a bounding box of volume $\frac32 \cdot 2^d$. The minimum curve volume $\mathit{MinCrvVol}$ is at least $\volume{E} - 2 = \volume{X} + \volume{Y} - 2 \geq 2^{d-2} + 2^{d-3} = \frac32 \cdot 2^{d-2}$.

\item[B3:] Given the bounds on $\volume{Y}$ and $\volume{X} \le 2^{d-1}$, Lemma~\ref{lem:prefixspan} tells us that the edges of $X$ have axes from $|\perm_x(1)|,\ldots,|\perm_x(d-1)|$, and the edges of $Y$ have axes from $|\perm_y(1)|,\ldots,|\perm_y(d-3)|$. Now let $a=|\perm_x(d)|$. By Theorem~\ref{thm:hyperorthogonal}, $\permdp(\perm_y,a) \leq \permdp(\perm_x,a) + 1 = 1$ and therefore $a$ is not included in $|\perm_y(1)|,\ldots,|\perm_y(d-3)|$. If $a = c$, it follows that $X$ and $Y$ lie in half-cubes that together constitute a unit cube of volume $2^d$, as in case B1.

    Otherwise, if $a \neq c$, it follows that $E$ may contain multiple edges of direction $c$ but does not include any edge with direction $a$. Therefore $E$ lies completely in a box that spans two 1-cubes in dimension $c$, half a 1-cube in dimension $a$, and one 1-cube in the remaining dimensions. The volume of this box is $2^d$.

\enlargethispage\baselineskip
\item[C:] By Lemma~\ref{lem:fitscubes}, each set of $\volume{E} - 1$ edges of $A_k$ is contained in a unit cube of $\lceil\log_2(\volume{E}-1)\rceil + 1 = \lfloor \log_2(\volume{E}-2)\rfloor + 2$ dimensions, of volume at most \hbox{$4(\volume{E}-2)$}.
\end{enumerate}
\end{proof}

\section{General construction method in three and more dimensions}\label{sec:existence}

\subsection{Extended curves and local edge distance}

\begin{wrapfigure}[22]{r}{0.5\textwidth}
\centering
\vspace{-1cm} 
{\hfill \includegraphics[scale=0.9]{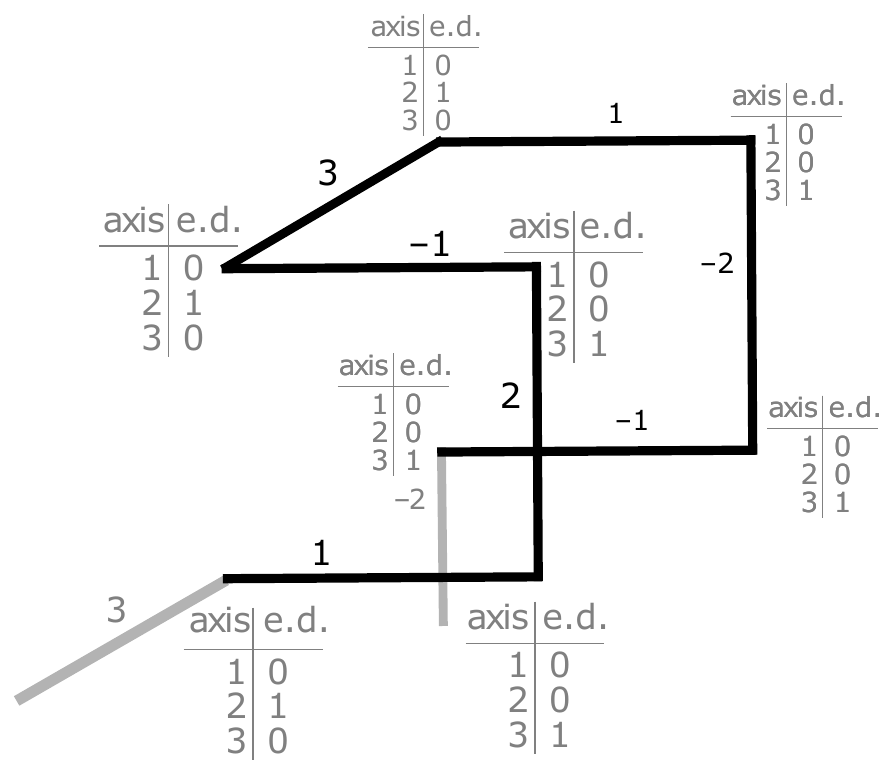} }
\vspace*{-0.5cm}
\caption{in black: $G(3)$ with the directions of its edges; in grey: an extension of $G(3)$ with an entry edge $\langle d\rangle$ and an exit edge $\langle -(d-1)\rangle$, with the edge distance table  for each vertex according to Definition~\ref{def:edgedistance}.}
\label{fig:G3ab1}
\end{wrapfigure}

In Section~\ref{sec:wellfolded}, Theorem~\ref{thm:wellfolded}, we learned about sufficient and necessary conditions for well-folded curves in general, and in Section~\ref{sec:hyperorthogonal}, Theorem~\ref{thm:hyperorthogonal}, we learned about specific conditions for hyperorthogonal well-folded curves.
It remains to show that curves satisfying both the general and the specific conditions actually exist.

In this section we will combine the conditions of Theorems \ref{thm:wellfolded} and~\ref{thm:hyperorthogonal} to derive conditions on the entry and exit points and the isometries used in the construction of hyperorthogonal well-folded curves. We will show how to construct curves that satisfy all conditions, for any $d \geq 3$ (recall that for $d = 2$, we have Hilbert's curve).

\begin{definition}\label{def:edgedistance}
The \emph{edge distance} of the axis $a \in \unsigneddim$ to the vertex $v$ within the curve $C$, denoted $\edgedist(C,v,a)$, is the distance along $C$ between $v$ and the closest edge with axis $a$; more precisely, $\edgedist(C,v,a)$ is one less than the length of the smallest subcurve of $C$ that includes $v$ and an edge with axis $a$.
(For a small example, see Figure~\ref{fig:G3ab1}.)
\end{definition}
Theorem~\ref{thm:hyperorthogonal} has a remarkable consequence:
\begin{lemma}\label{lem:edgedistance}
Let $A_0, \ldots A_{k+1}$ be a sequence of well-folded hyperorthogonal
curves constructed by inflation. Then we have, for all axes $a \in
\{1,\ldots,d\}$ and all vertices $v_i$ of $A_k$, $\permdp(\perm_i, a)
\leq \edgedist(A_k,v_i,a)$.
\end{lemma}
\begin{proof}
The proof goes by induction on increasing edge distance.
If $\edgedist(A_k,v_i,a) = 0$, then $a \in\{|e_i|,|e_{i-1}|\}$, and, by the first condition of Theorem~\ref{thm:hyperorthogonal}, we have $\permdp(\perm_i,a) = 0 = \edgedist(A_k,v_i,a)$.

Now suppose $\edgedist(A_k,v_i,a) > 0$. Then we can choose $j \in \{i-1,i+1\}$ such that $\edgedist(A_k,v_j,a) = \edgedist(A_k,v_i,a) - 1$, and by induction we can assume $\permdp(\perm_j, a) \leq \edgedist(A_k,v_j,a)$. Then it follows from the second condition of Theorem~\ref{thm:hyperorthogonal} that we have $\permdp(\perm_i,a) \leq \permdp(\perm_j,a) + 1 \leq \edgedist(A_k,v_j,a) + 1 = \edgedist(A_k,v_i,a)$.
\end{proof}

Lemma~\ref{lem:edgedistance} gives us the following idea for an algorithm to specify the permutations $|\perm_i|$, except for the order of the last two elements: simply sort all axes $a \in \unsigneddim$ by order of decreasing edge distance $\edgedist(A_k,v_i,a)$. In fact, as we will show, a version of this algorithm suffices, that only considers edge distances within small subcurves. For this purpose we define the notion of \emph{extended} curves, which can be seen as curves together with an indication of how the curve is connected to preceding an succeeding curves:

\begin{definition}\label{def:extended}
An \emph{extended curve} is a curve that is extended with an \emph{entry edge} leading to the first vertex (the entry point) and an \emph{exit edge} originating from the last vertex (the exit point). The origin of the entry edge and the destination of the exit edge are not considered to be part of the curve.
\end{definition}

We use prime symbols to distinguish extended curves from non-extended curves: when $B$ is a curve, we may use $B'$ to denote a particular extension of $B$, and when $B'$ is an extended curve, we use $B$ te denote the curve without the extensions. Note that by our definition, $B$ and $B'$ always have the same vertices; they only differ in the number of edges. In particular, if we extend an approximating curve $A_k$ that has edges $\langle e_{k,1},\ldots,e_{k,K-1}\rangle$, we denote the extended curve by $A'_k$, the entry edge by $e_{k,0}$ and the exit edge by $e_{k,K}$. If $B$ is a subcurve of $A'$, then the entry edge of $B'$ is the edge that leads to the entry vertex of $B$ in $A'$, and the exit edge of $B'$ is the edge that originates from the exit vertex of $B$ in $A'$. In particular, the extended child curve $C'_{k,i}$ of $A'_k$ would be $C_{k,i}$ extended with entry edge $\langle e_{k,i-1}\rangle$ and exit edge $\langle e_{k,i}\rangle$.

The definition of well-foldedness (Definition~\ref{def:wellfolded}) can be applied to extended curves, with the base case that an extended curve that consists of only an entry edge, a single vertex, and an exit edge, is well-folded.
The conditions for well-foldedness from Theorem~\ref{thm:wellfolded} are applicable as well. In that case it is natural to require that we would obtain a valid curve if we would add the origin of the entry edge and the destination of the exit edge to the curve. This can be ensured as follows: we take the entry edge into account by extending the second condition of Theorem~\ref{thm:wellfolded} to the case $i = 0$; given the first condition, the second condition can also be written as: $\sgn\big(\perm^{-1}_{i}(e_i)\big) = 1$ if and only if $|e_i| = |\perm_i(d)|$, and we take the exit edge into account by extending this form of the condition to the case $i = K$.

The definition and conditions of hyperorthogonality (Definition~\ref{def:hyperorthogonal} and Theorem~\ref{thm:hyperorthogonal}) can be applied to extended curves, if, in condition~1 of Theorem~\ref{thm:hyperorthogonal}, we also take the entry and exit edge into account. Concretely, this means condition~1 should be extended with $\permdp(\perm_{1}, e_0) = 0$ and $\permdp(\perm_K, e_K) = 0$. The definition of edge distance (Definition~\ref{def:edgedistance}) and its relation to hyperorthogonality (Lemma~\ref{lem:edgedistance}) can now be applied directly to extended curves.

We can now define a version of edge distance that only considers small subcurves:

\begin{definition}\label{def:localedgedistance}
Let $A'_k$ be an extended well-folded curve obtained by inflation from $A'_{k-1}$. Let $v$ be a vertex of $A_k$, let $a \in \unsigneddim$ be any axis, let the subcurve $C$ of $A_k$ be the child curve of $A'_{k-1}$ that contains $v$, and let $C'$ be the extension of $C$ within $A'_k$. We define the \emph{local} edge distance of the axis $a$ to the vertex $v$ within the curve $A'_k$, denoted $\localedgedist(A'_k,v,a)$, as $\edgedist(C',v,a)$.
\end{definition}

\subsection{Hyperorthogonal curves from inflation of extended curves}\label{sec:generalconstruction}

\begin{lemma}\label{lem:edgedistalgo}
Suppose we construct a sequence of extended well-folded curves $A'_0,A'_1,\ldots$ by inflation such that the elements of
each permutation $|\perm_{k,i}|$ are sorted by order of decreasing local edge distance to $v_{k,i}$ in $A'_k$. Then these permutations satisfy conditions 1 and 2 of Theorem~\ref{thm:hyperorthogonal}.
\end{lemma}

\begin{proof}
The proof goes by induction on increasing $k$. As the base case we take $k = 0$, and observe that $A'_0$, which contains only a single vertex and two edges, trivially satisfies Theorem~\ref{thm:hyperorthogonal}.
Now suppose $k > 0$ and the permutations associated with the vertices of $A'_{k-1}$ satisfy the conditions of Theorem~\ref{thm:hyperorthogonal}. We will now show that, if we choose the permutations $\perm_1,\perm_2,\ldots$ associated with the vertices $v_1,v_2,\ldots$ of $A'_k$ in such a way that the elements of $|\perm_i|$ are sorted by order of decreasing local edge distance to $v_i$ in $A'_k$, then these permutations satisfy the conditions of Theorem~\ref{thm:hyperorthogonal} as well.

Consider any vertex $v_i$ in $A'_k$. Let $C'$ be the extended child curve of $A'_{k-1}$ that contains $v_i$, and let $\tau$ be the signed permutation such that $C$ is a translate of $\tau(G(d)$. Since $C'$ includes both edges of $A'_k$ that are incident on $v_i$, we have $\localedgedist(A'_k,v_i,a) = 0$ if and only if $a \in \{|e_{i-1}|,|e_i|\}$. Hence, these two axes will be placed at the last positions within $|\perm_i|$, so that $\permdp(\perm_i,a) = 0$, and condition~1 of Theorem~\ref{thm:hyperorthogonal} is satisfied.

For condition 2, observe that in $C$, being a transformation of $G(d)$, the edges with axis $|\tau(1)|$ and edges with other axes alternate, starting and ending with an edge with axis $|\tau(1)|$. By the induction hypothesis, $A'_{k-1}$ satisfies the conditions of Theorem~\ref{thm:hyperorthogonal}, which implies that the edges immediately preceding and following $C$ in $A'_k$ have axes with depth zero in $\tau$. Therefore these axes differ from $|\tau(1)|$, which has depth $d-2$ (remember that this section is concerned with $d$-dimensional curves for $d \geq 3$). Thus, also in $C'$ edges with axis $|\tau(1)|$ and edges with other axes alternate.

Now suppose, for the sake of contradiction, that there are two axes $a \neq b$, both different from $|\tau(1)|$, such that $\localedgedist(A'_k,v_i,a) = \localedgedist(A'_k,v_i,b)$ for some $v_i\in C'$. Then $C'$ must contain an edge sequence of even length, more precisely $2\ast \localedgedist(A'_k,v_i,a) + 2$, with $v_i$ in the middle, that starts with an edge with axis $a$ and ends with an edge with axis $b$. However, this contradicts the fact that edges with axis $|\tau(1)|$ and edges with other axes alternate. Hence, apart from a pair of axes with local edge distance zero (among which $|\tau(1)|$), no pair of axes has the same local edge distance.

If we increase $i$ by one while staying inside the same child curve $C$, each local edge distance changes by at most one, and therefore each axis can move up or down in the order of $\tau$ by at most one position. This establishes condition 2 of Theorem~\ref{thm:hyperorthogonal} as long as we stay in the same child curve, that is, as long as $\lceil i/2^d\rceil$ does not change, that is, for all $i \neq 0 \pmod{2^d}$.

If $i = 0 \pmod{2^d}$, we need to take more care, as $v_i$ lies in a child curve $C_j$ of $A'_{k-1}$ while $v_{i+1}$ lies in the next child curve $C_{j+1}$. 
Now, since $A'_{k-1}$ satisfies condition 1 of Theorem~\ref{thm:hyperorthogonal}, there must be $g, h \in \{d-1,d\}$ such that $|\perm_{k-1,j}(g)| = |\perm_{k-1,j+1}(h)| = |e_{k-1,j}|$ where $e_{k-1,j}$ is the direction $e_{k,i}$ of the edge $(v_i,v_{i+1})$ in $A'_k$. Define $g', h' \in \{d-1,d\}$ by $g' \neq g$ and $h' \neq h$.\\
Now, for $v_i$, sorting by decreasing edge distance within $C'_j$ results in $|\perm_{k,i}| =$\\$\big[|\perm_{k-1,j}(g')|,|\perm_{k-1,j}(d-2)|,|\perm_{k-1,j}(d-3)|,\ldots,|\perm_{k-1,j}(2)|,|\perm_{k-1,j}(1)|,|\perm_{k-1,j}(g)|\big]$,\\ where the order of the last two elements is undetermined, and likewise\\for $v_{i+1}$, sorting by decreasing edge distance within $C'_{j+1}$ results in $|\perm_{k,i+1}| =$\\$\big[|\perm_{k-1,j+1}(h')|,|\perm_{k-1,j+1}(d-2)|,|\perm_{k-1,j+1}(d-3)|,\ldots,|\perm_{k-1,j+1}(2)|,|\perm_{k-1,j+1}(1)|,|\perm_{k-1,j+1}(h)|\big]$,\\ where also the order of the last two elements is undetermined.
For $a = |e_{k-1,j}|$ we thus have $\permdp(\perm_{k,i}(a)) = \permdp(\perm_{k,i+1}(a)) = 0$, and for $a \neq |e_{k-1,j}|$ we have $\permdp(\perm_{k,i}(a)) = d - 2 - \permdp(\perm_{k-1,j}(a))$ and $\permdp(\perm_{k,i+1}(a)) = d - 2 - \permdp(\perm_{k-1,j+1}(a))$. By the induction hypothesis, $\perm_{k-1,j}$ and $\perm_{k-1,j+1}$ satisfy condition 2 of Theorem~\ref{thm:hyperorthogonal} and therefore we have $|\permdp(\perm_{k-1,j}(a)) - \permdp(\perm_{k-1,j+1}(a))| \leq 1$, and hence $|\permdp(\perm_{k,i}(a)) - \permdp(\perm_{k,i+1}(a))| \leq 1$, which establishes condition 2 of Theorem~\ref{thm:hyperorthogonal}.
\end{proof}

The above lemma still leaves the order of the last two elements of each $|\perm_i|$ undetermined, since these are always the two axes with edge distance zero. To prove that hyperorthogonal well-folded curves exist, it now suffices to show that we can order the last two elements and choose the signs of each $\perm_i$ such that the conditions of Theorem~\ref{thm:wellfolded} are satisfied. We obtain:

\begin{theorem}\label{thm:existence}
For each choice of an entry direction $e_{0,0}$ and an exit direction $e_{0,1}$ and for each choice for the signs of $\perm^{-1}_{k,1}(j)$ for all $k$ and $j$ such that $\sgn(\perm^{-1}_{k,1}(e_{k,0})) = 1$ for all $k$, there is a unique hyperorthogonal, well-folded space-filling curve $f$ approximated by $A'_0,A'_1,\ldots$, where each curve $A'_k$ with $k > 0$ is constructed by inflation from $A'_{k-1}$ and the elements of each permutation $|\perm_{k,i}|$ are sorted by order of decreasing local edge distance to $v_i$ in $A'_k$.
\end{theorem}

\begin{proof}
For each level $k$, we generate $A'_k$ as follows. We loop over all $i \in \{1,\ldots,K-1\}$, where $K = \twodk$, and proceed as follows. The conditions of Theorem~\ref{thm:wellfolded} require $\sgn(\perm^{-1}_{i+1}(e_i)) = 1$. We now choose $|\perm_i(d)|$ such that $|\perm_i(d)| = |e_i|$ if and only if $\sgn(\perm^{-1}_i(e_i)) = 1$: this is always possible since $|e_i|$ is among the last two elements of $|\perm_i|$ whose order was undetermined. Thus we satisfy the second condition of Theorem~\ref{thm:wellfolded} for $j = |e_i|$. With $|\perm_i|$ completely determined, we can now fill in the remaining signs of $\perm_{i+1}$ such that they fulfill the first condition of Theorem~\ref{thm:wellfolded}. Finally, we determine $|\perm_K(d)|$ as dictated by the exit direction $e_K$ in the same way as we determined $|\perm_i(d)|$ for $i < K$.
\end{proof}

Note that if $A'_0,A'_1,\ldots$ is a set of extended hyperorthogonal well-folded curves constructed by inflation, then they are approximating curves of a hyperorthogonal well-folded space-filling curve. Note, however, that not every hyperorthogonal well-folded space-filling curve can be described by such a set $A'_0,A'_1,\ldots$. There can also be hyperorthogonal well-folded space-filling curves that can be described by a set of \emph{non-extended} hyperorthogonal well-folded approximating curves $A_0,A_1,\ldots$ which cannot be extended to a set $A'_0,A'_1,\ldots$ of hyperorthogonal curves in such a way that each curve $A'_k$ is equal to $A_k$ extended with an entry edge $\langle e_{0,0}\rangle$ and an exit edge $\langle e_{0,1}\rangle$. Examples would include symmetric space-filling curves, closed space-filling curves (that is, curves that start end end in the same point), and space-filling curves that start in the interior of the unit cube.

\section{Self-similar curves in three and more dimensions}\label{sec:selfsimilar}

\subsection{The challenge}

By Observation~\ref{obs:startingpoint}, a choice of signs of $\perm^{-1}_{k,1}(j)$ for all $k$ and $j$ specifies the starting point $f(0)$ of the space-filling curve $f$ in Theorem~\ref{thm:existence}. Thus, the proof of Theorem~\ref{thm:existence} is a constructive proof that a hyperorthogonal, well-folded space-filling curve exists for any choice of starting point on the boundary of the unit hypercube.

In a practical setting, such as described in Section~\ref{sec:rtrees}, one may want to sort points in the order in which they appear along the curve. To this end we need a comparison operator that decides which of any two given points $p$ and $q$ comes first along the curve. We can do so by determining the largest $k$ such that there is a hypercube $H_{k,i}$, corresponding to a vertex $v_{k,i}$ of $A_k$, which contains both points. Then we can use $\perm_{k,i}$ to determine in which order the $2^d$ subcubes of this hypercube are traversed, and in particular, in which order this traversal visits the two subcubes containing $p$ and $q$. The efficiency of the comparison operator now depends on how efficiently we can determine $\perm_{k,i}$ for any $k$ and $i$. Unfortunately, straightforward application of Theorem~\ref{thm:existence} would require us to derive $\perm_{k,i}$ in an incremental fashion that explicitly constructs all $\perm_{k,j}$ for all $j < i$. In practice we will need a less time-consuming way to derive $\perm_{k,i}$. This seems rather difficult to achieve if we allow ourselves to choose the signs of the permutations $\perm_{k,1}$ arbitrarily.


To enable us to determine a permutation $\perm_{k,i}$ more efficiently, we will, in this section, restrict the curves to be \emph{self-similar}, that is, any approximating curve $A_{k+1}$ is the concatenation of $2^d$ isometric copies of~$A_k$. Recall that taking the reverse is also an isometric mapping.

\begin{figure}
\centering
\includegraphics{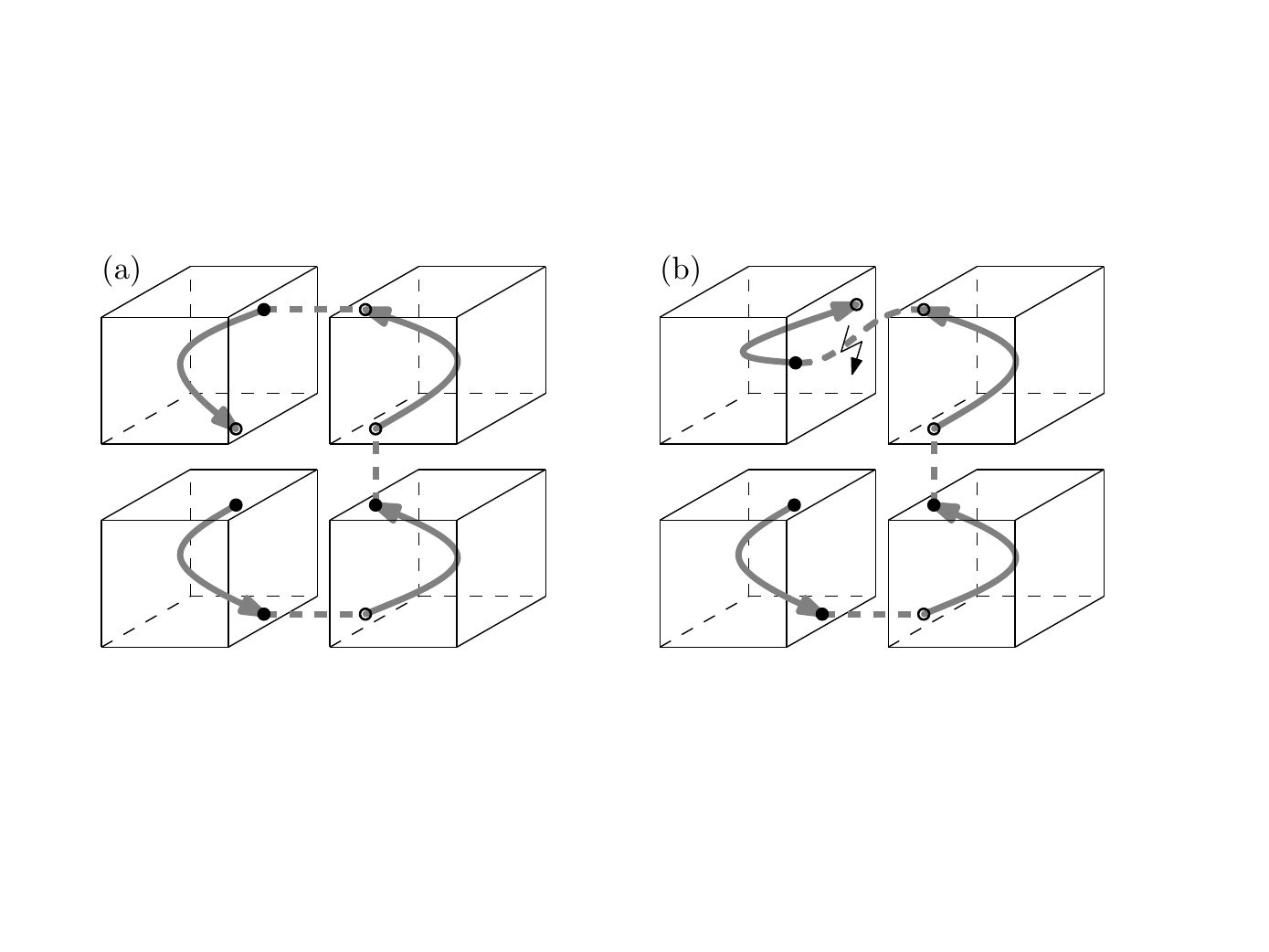}
\caption{Each cube in this figure shows a grey curve serving as a very rough sketch of $A_k$, with its entry and exit on the interior of two non-opposite $(d-1)$-dimensional faces of the cube. (a) It is easy to connect up four copies of such a curve. (b) Assembling more than four copies requires rotating the cubes to bend the path into other directions. In general, the rotations will break the connections between one copy and the next.}
\label{fig:mismatch}
\end{figure}

Note that for $d = 2$, Hilbert's original curve is the only self-similar well-folded curve. So for the purposes of Sections \ref{sec:extended1curves} to~\ref{sec:selfsimilarconstruction} we will assume $d \geq 3$.

In Section~\ref{sec:extended1curves} we find that all self-similar hyperorthogonal well-folded space-filling curves fit the framework of Theorem~\ref{thm:existence}, that is, they can be described by a series of \emph{extended} approximating curves. Moreover, we find that the only extended 1-curves that are relevant for the study of self-similar hyperorthogonal well-folded space-filling curves are isometries of one particular extension of $G(d)$. As shown in Section~\ref{sec:generalconstruction}, a (not necesssarily self-similar) hyperorthogonal well-folded space-filling curve does not need to start at a corner of the unit cube, but may start on the interior of a $(d-1)$-dimensional face. In Section~\ref{sec:relativecoordinates} we set up notation which is helpful in distinguishing different possible locations for entries and exits. In Section~\ref{sec:entrytoexit} we analyse how the choice of the entry of $C_{1,1}$ propagates to the other child curves $C_{1,2},\ldots,C_{1,D}$ of $A_1$. We will find that if the entry point of $A_1$ lies on the interior of a $(d-1)$-dimensional face, the exit point will also lie on the interior of a face. If the entry and exit of an approximating curve $A_k$ indeed lie on different but non-opposite faces, it is now trivial to connect up four copies of $A_k$ to make a cycle using only reflection and reversal transformations, see Figure~\ref{fig:mismatch}(a).

However, to get beyond four copies of $A_k$ and assemble $2^d$ copies to make a full $A_{k+1}$, we need to rotate some copies of $A_k$ in various ways to bend the path into all $d$ dimensions. The difficulty is to ensure that despite the various rotations, the connection points on the faces will still match up. This will not automatically be the case, see Figure~\ref{fig:mismatch}(b). As we will find in Section~\ref{sec:possibleentry}, this forces most of the coordinates of the entry and exit points to be equal, so that these points become invariant under the necessary transformations.

This strongly restricts the possible shapes of self-similar hyperorthogonal well-folded curves, but not too much: in Section~\ref{sec:selfsimilarconstruction}, we find that such curves do in fact exist. It turns out that for any $d \geq 3$, only two different starting points (modulo rotation and reflection) exist for such curves.

\subsection{Extensions in self-similar curves}\label{sec:extended1curves}

As noted in Section~\ref{sec:generalconstruction}, not all hyperorthogonal, well-folded space-filling curves may be approximated by \emph{extended} hyperorthogonal, well-folded curves. However, for self-similar curves this can always be done. But before proving this, we will first have a look at extensions of $G(d)$. In extended hyperorthogonal well-folded approximating curves, we find only one particular extension of $G(d)$:

\begin{definition}\label{def:extendedGraycode}
Let $G'(d)$ be the concatenation of $\langle d\rangle$, $G(d)$, and $\langle -(d-1)\rangle$.
\end{definition}

\begin{lemma}\label{lem:types}
Let $A'_0,\ldots,A'_{k+1}$ be a sequence of extended hyperorthogonal well-folded curves constructed by inflation. Then each extended child curve $C'_{k,i}$ is the image of an isometry of $G'(d)$.
\end{lemma}

\begin{proof}
Consider an extended child curve $C'_i$ of $A'_k$.
Since $A'_k$ is hyperorthogonal, by Theorem~\ref{thm:hyperorthogonal} the axes of $e_{i-1}$ and $e_i$ in $C'_i$ must be $|\perm_i(d-1)|$ and $|\perm_i(d)|$. This leaves two possibilities for matching axes to edges. The first possibility is to put the edge with axis $|\perm_i(d)|$ at the beginning and the edge with axis $|\perm_i(d-1)|$ at the end; the signs of the directions follow from the fact that the vertex preceding $C_i$ in $A_{k+1}$ and the vertex following $C_i$ in $A_{k+1}$ must lie outside $C_i$. The second possibility is to put the edge with axis $|\perm_i(d-1)|$ at the beginning and the edge with axis $|\perm_i(d)|$ at the end, so that we obtain a concatenation of $\langle \perm_i(d-1)\rangle$, $\perm_i(G(d))$ and $\langle \perm_i(d)\rangle$, which is the reverse of $\perm_i(G'(d))$ with reflection in coordinate $\perm_i(d)$.
\end{proof}

Note that the proof of Lemma~\ref{lem:types} does not require the approximated space-filling curve to be self-similar.

\begin{definition}\label{def:types}
If, in the above lemma, the isometry that maps $G'(d)$ to $C'_i$ is composed exclusively of rotation, reflection, and translation, then we say $C'_{k,i}$ is of type 0; otherwise, that is, if the isometry that maps $G'(d)$ to $C'_{k,i}$ involves reversing the curve, then we say $C'_{k,i}$ is of type 1. We denote the type of $C'_{k,i}$ by $T_{k,i}$.
\end{definition}
As always, when the first subscript to $T$ is clear from the context, or when a statement holds for any value of the first subscript, we may omit the subscript.

In the following observation we use the Iverson bracket notation: when $P$ is a expression that evaluates to true or false, then $\big[P\big] = 0$ if $P$ is false, and $\big[P\big] = 1$ if $P$ is true. The observation is the following: the type of an isometric image $C'_i$ of $G'(d)$ is zero if and only if $|\sigma_i(d)|$ is the axis of the entry edge $e_{i-1}$, and the type is one if and only if $|\sigma_i(d)|$ is the axis of the exit edge $e_{i}$. In other words:

\begin{observation}\label{obs:typeandorientation}
$T_i = \big[|\sigma_i(d)| \neq |e_{i-1}|\big] = \big[|\sigma_i(d)| = |e_{i}|\big]$.
\end{observation}

We can now prove that all \emph{self-similar} hyperorthogonal, well-folded space-filling curves can be approximated by a series of \emph{extended} hyperorthogonal, well-folded curves:

\begin{lemma}\label{lem:selfsimilartypes}
Let $f$ be a self-similar, hyperorthogonal, well-folded space-filling curve. Then an isometric copy of $f$ is approximated by a series of extended hyperorthogonal, well-folded curves $A'_0,A'_1,\ldots$, where $A'_0 = \langle d, -(d-1)\rangle$, $\sigma_{0,1} = [1,\ldots,d]$, and each extended curve $A'_k$ with $k > 0$ is obtained by inflation from $A'_{k-1}$. \end{lemma}
\begin{proof}
For any $k \geq 0$ and $1 < i < 2^d$, let $B_{1,i,k}$ be the $k$-curve that is a subcurve of the non-extended approximating curve $A_{k+1}$ and results from $k$ steps of inflation of the vertex $v_i$ of $A_1$. Since $A_{k+1}$ is hyperorthogonal and well-folded, the extended curve $B'_{1,i,k}$ that consists of $B_{1,i,k}$ with entry edge $\langle e_{i-1}\rangle$ and exit edge $\langle e_i\rangle$ must also be hyperorthogonal and well-folded. It follows that $B'_{1,i,0},B'_{1,i,1},B'_{1,i,2},\ldots$ is a sequence of extended hyperorthogonal, well-folded curves that approximate the space-filling curve $f_i$ which consists of $f$ restricted to the hypercube $H_i$ that corresponds to $v_i$. By Lemma~\ref{lem:types}, it follows that there is an isometric transformation that maps $B'_{1,i,1}$ to $G'(d)$, and thus, $B'_{1,i,0}$ to $\langle d, -(d-1)\rangle$.

Because $f$ is self-similar, the same series of curves that approximates $f_i$ also approximates $f$, up to isometric transformations.
\end{proof}

\begin{lemma}\label{lem:asymmetric}
Let $F'$ be an extended hyperorthogonal well-folded curve obtained by one step of inflation from $G'(d)$, and let $R'$ be an extended hyperorthogonal well-folded curve obtained by one step of inflation from $\reverse{G'(d)}$. Then no non-reverse isometry of $R'$ can visit its vertices in the same order as $F'$.
\end{lemma}
\begin{proof}
Suppose $\tau$ is a non-reverse isometry, expressed by a signed permutation, such that $\tau(R')$ visits its vertices in the same order as $F'$, and hence the axes of the edges of $\tau(R')$ and $F'$ are the same, apart from, possibly the entry and the exit edge. Then $\tau\big(\reverse{G'(d)}\big)$ must also visit its vertices in the same order as $G'(d)$, so $\tau = [1,\ldots,d-1,-d]$, and $|\tau|$ is the identity permutation. However, since the entry edge of $\tau(R')$ has axis $d-1$ while the entry edge of $F'$ has axis $d$, the axes of the edges of $\tau(R')$ and $F'$ must differ in the first child curve of $\tau\big(\reverse{G'(d)}\big)$ and $G'(d)$, respectively.
\end{proof}

Note that another way to put the last line of the lemma is to say that any non-reverse isometry of $F'$ must differ from $R'$ in more than just the entry and/or exit edge.

\begin{corollary}\label{cor:asymmetric}
For $d \geq 3$, any $d$-dimensional self-similar hyperorthogonal well-folded space-filling curve is asymmetric.
\end{corollary}

\subsection{Relative coordinates of entries and exits}\label{sec:relativecoordinates}

In the following subsections, the following notation will be helpful.

\begin{definition}\label{def:ent_ext}
Let $\absentr_{k,m}, \absexit_{k,m} : \unsigneddim \rightarrow \{0,\ldots,2^{k+1}-1\}$ be functions that give the coordinates of the entry and exit point of $C_{k,m}$,
that is, the entry point of $C_{k,m}$ has coordinates $\big(\absentr_m(1),\ldots,\absentr_m(d)\big)$ and the exit point has coordinates $\big(\absexit_m(1),\ldots,\absexit_m(d)\big)$.
\end{definition}
Note that $C_{k,m}$ is a 1-curve that is a subcurve of $A_{k+1}$, we have $\absentr_{k,m}(j) = \flipped(\perm^{-1}_{k,m}(j)) \pmod 2$, and $\absexit_{k,m}(j) = \absentr_{k,m}(j) \pmod 2$ if and only if $|\perm^{-1}_{k,m}(j)| \neq d$.

\begin{definition}\label{def:relativecoordinates}
The \emph{relative coordinate vector} of a vertex $v$ is the vector $r$ such that $r[j] = 0$ if $v[j] \bmod 4 \in \{0,3\}$, and $r[j] = 1$ if $v[j] \bmod 4 \in \{1,2\}$.
\end{definition}
The relative coordinates of a vertex $v_n$ of $A_{k+1}$ tell us, for each dimension, whether the vertex is on the outside (0) or on the inside (1) with respect to the 2-cube that results from inflating the inflation of the vertex $v_j$ of $A_{k-1}$, where $j = \lceil n/D^2\rceil$.

\begin{definition}\label{def:rlent_rlext}
Let $\relentr_{k,m}, \relexit_{k,m}: \unsigneddim \rightarrow \{0,1\}$ be functions that give us the \emph{relative} coordinates of the entry and exit point of $C_{k,m}$.
\end{definition}

\begin{observation}\label{obs:relationabsandrel}\leavevmode\\
$\relentr_{k,m}(j) = \big(\absentr_{k,m}(j) + v_{k,m}[j]\big) \bmod 2$, or equivalently,\\
$\relentr_{k,m}(j) = \big(\flipped(\perm^{-1}_{k,m}(j)) + v_{k,m}[j]\big) \bmod 2$, and\\
$\relexit_{k,m}(j) = \big(\absexit_{k,m}(j) + v_{k,m}[j]\big) \bmod 2$.
\end{observation}
Note that in the above observation, $\absentr_{k,m}(j)$ is a coordinate of the entry of $C_{k,m}$, which is a vertex of $A_{k+1}$, while $v_{k,m}[j]$ is a coordinate of a vertex of $A_k$. In fact, $v_{k,m}[j] = \lfloor\absentr_{k,m}(j) / 2\rfloor$ and $\absentr_{k,m}(j)$ must be either $2 \ast v_{k,m}[j]$ or $2 \ast v_{k,m}[j] + 1$. Thus, if $\relentr_{k,m}$ and $v_{k,m}$ are given, this determines $\absentr_{k,m}$ and hence, the signs of $\perm^{-1}_{k,m}$.

\begin{observation}\label{obs:relationrelentryandexit}\leavevmode\\
$\relexit_{k,m}(j) = \relentr_{k,m}(j) + \big[|\perm_{k,m}(d)| = j\big] \pmod 2$;\\
for $m < D^k$ we have $\relexit_{k,m} = \relentr_{k,m+1}$.
\end{observation}

As always, when the first subscript to $\absentr$, $\absexit$, $\relentr$ or $\relexit$ is clear from the context, or when a statement holds for any value of the first subscript, we may omit the subscript.

\subsection{Relation between entry and exit of a 2-curve}\label{sec:entrytoexit}

A direct consequence of Lemma~\ref{lem:selfsimilartypes} is that for a self-similar curve we may assume, without loss of generality (modulo reflection, rotation and reversal), that $A'_1 = C'_{0,1} = G'(d)$, so with type $T_{0,1} = 0$, entry edge $\langle d\rangle$ and exit edge $\langle -(d-1)\rangle$. Moreover, in $A_2$, we should have $v_1[d] = 0$ and $v_K[d-1] = 0$, where $K = D^2 = (2^d)^2$, so that the child curves $C_{1,1}$ and $C_{1,D}$ in $A_2$ can be extended with, respectively, the same entry edge $\langle d\rangle$ and the same exit edge $\langle -(d-1)\rangle$ as $A_1$. Note that we can rewrite the conditions on $v_1[d]$ and $v_K[d-1]$ as $\relentr_{1,1}(d) = 0$ and $\relexit_{1,D}(d-1) = 0$.

In this subsection we consider extended hyperorthogonal well-folded approximating curves $A'_1$ and $A'_2$ that fulfill these basic conditions, that is, $A'_1 = G'(d)$, $\relentr_{1,1}(d) = 0$ and $\relexit_{1,D}(d) = 0$, without assuming, at this point, that $A'_1$ and $A'_2$ are indeed approximations of a self-similar space-filling curve. In particular, in this subsection we analyse how the choice of the entry point of $C_{1,1}$ propagates to the other child curves $C_{1,2},\ldots,C_{1,D}$ of $A'_1$. Because the whole subsection focuses on the child curves of $A'_1$ that constitute $A'_2$, we will omit the first subscripts to $C$, $\perm$, $T$, $\relentr$ and $\relexit$: they would always be 1.

\begin{definition}\label{def:omega}
Let $\omega$ be the permutation $[d-1, 2, \ldots, d-2, d, 1]$.
\end{definition}

By tracing the relative coordinates of the entry and exit points through the child curves of $A'_1$ that make up $A'_2$, using the conditions of Theorems \ref{thm:wellfolded} and~\ref{thm:hyperorthogonal}, we find $\relexit_D = \relentr_1 \compose \omega$ (in Lemma \ref{lem:omega}). To prove this we need the following three lemmas from which the proofs can be skipped at first reading.

\begin{lemma}\label{lem:opening}\leavevmode
\begin{itemize}
\item $\relentr_1(d) = 0$.
\item If $\relentr_1(1) = 0$,
then $|\perm_1(d-1)| = d$ and $|\perm_1(d)| = 1$,\\
otherwise $|\perm_1(d-1)| = 1$ and $|\perm_1(d)| = d$.
\item $T_1 = 1 - \relentr_1(1)$.
\end{itemize}
\end{lemma}

\begin{proof}
The first item, $\relentr_1(d) = 0$, follows from the fact that the entry edge is $\langle d\rangle$, by the assumptions of this subsection.
The second item follows from the fact that we need $\relexit_1(1)=1$ to be able to connect $C_1$ to $C_2$ with the first edge of $G(d)$, which has direction~$1$.
The third item follows from Observation~\ref{obs:typeandorientation} by $T_1 = \big[|\perm_1(d)| = |e_1|\big] = \big[|\perm_1(d)| = 1\big] = \big[\relentr_1(1) = 0\big] = 1 - \relentr_1(1)$.
\end{proof}

\begin{lemma}\label{lem:developmenteven}
For even $i < 2^d$ we have:
\begin{itemize}
\item $\relentr_i(1) = 1$; for $1 < j < d$, $\relentr_i(j) = \relentr_1(j)$; $\relentr_i(d) = \relentr_1(1)$.
\item If $\relentr_2(|e_i|) = 0$ then $|\perm_i(d-1)| = |\perm_{i+1}(d-1)| = 1$ and $|\perm_i(d)| = |\perm_{i+1}(d)| = |e_i|$,
otherwise $|\perm_i(d-1)| = |\perm_{i+1}(d-1)| = |e_i|$ and $|\perm_i(d)| = |\perm_{i+1}(d)| = 1$.
\item $T_{i+1} = \relentr_2(|e_i|)$; $T_i = 1 - T_{i+1}$.
\end{itemize}
\end{lemma}

\begin{proof}
Throughout this proof, all additions and subtractions are to be interpreted modulo 2.

We first handle the case $i = 2$.

It is straightforward to calculate the relative entry function $\relentr_2$ from $\relentr_1$ and $|\perm_1(d)|$ using Observation~\ref{obs:relationrelentryandexit}. In particular, with the second item of Lemma~\ref{lem:opening} we get $\relentr_2(1) = \relentr_1(1) + \big[|\perm_1(d)| = 1\big] = \relentr_1(1) + (1 - \relentr_1(1)) = 1$; for $1 < j < d$ we have $\relentr_2(j) = \relentr_1(j) + \big[|\perm_1(d)| = j\big] = \relentr_1(j)$; and $\relentr_2(d) = \relentr_1(d) + \big[|\perm_1(d)| = d\big] = 0 + \relentr_1(1)$.

Now, because $i$ is even, by Lemma~\ref{lem:graycodealternatingpattern} we have $|e_i| \neq 1$ and $|e_{i-1}| = |e_{i+1}| = 1$.
With respect to axis $|e_i|$, the exit point of $C_i$ must be on the inside of the 2-cube traversed by $A'_2$, otherwise it cannot be connected to the next child curve $C_{i+1}$ by an edge with axis $|e_i|$. In other words: we must have $\relexit_i(|e_i|) = 1$, and therefore, by Observation~\ref{obs:relationrelentryandexit}, $\big[|\perm_i(d)| = |e_i|\big] = 1 - \relentr_i(|e_i|)$.\\
Therefore, if $\relentr_i(|e_i|) = 0$, then $|\perm_i(d)| = |e_i|$ and therefore, by Theorem~\ref{thm:hyperorthogonal},\break $|\perm_i(d-1)| = |e_{i-1}| = 1$; moreover, $\relentr_{i+1}(1) = \relentr_i(1) + \big[|\perm_i(d)| = 1\big] = \relentr_i(1) = 1$, so $\big[|\perm_{i+1}(d)| = 1\big] = 1 - \relentr_{i+1}(1) = 0$, in other words, $|\perm_{i+1}(d)| \neq 1 = |e_{i+1}|$, and, by Theorem~\ref{thm:hyperorthogonal}, $|\perm_{i+1}(d)| = |e_i|$ and $|\perm_{i+1}(d-1)| = 1$.\\
Otherwise, if $\relentr_i(|e_i|) = 1$, then $|\perm_i(d)| \neq |e_i|$, so, by Theorem~\ref{thm:hyperorthogonal}, $|\perm_i(d)| = 1$ and $|\perm_i(d-1)| = |e_i|$; moreover, $\relentr_{i+1}(1) = \relentr_i(1) + \big[|\perm_i(d)| = 1\big] = \relentr_i(1) + 1 = 0$, so $\big[|\perm_{i+1}(d)| = 1\big] = 1 - \relentr_{i+1}(1) = 1$, in other words, $|\perm_{i+1}(d)| = 1 = |e_{i+1}|$, and, by Theorem~\ref{thm:hyperorthogonal}, $|\perm_{i+1}(d-1)| = |e_i|$.

The third item of the lemma now follows from Observation~\ref{obs:typeandorientation}:
$T_{i+1} = \big[|\perm_{i+1}(d)| \neq |e_i|\big] = \relentr_i(|e_i|)$, and
$T_i = \big[|\perm_i(d)| = |e_i|\big] = 1 - \relentr_i(|e_i|) = 1 - T_{i+1}$.

The cases $i > 2$, for even $i$, then follow by induction, using that $\perm_{i+1}(d) = \perm_i(d)$, so that, by Observation~\ref{obs:relationrelentryandexit}, $\relentr_{i+2} = \relexit_{i+1} = \relentr_i$.
\end{proof}

\begin{lemma}\label{lem:finale}\leavevmode
\begin{itemize}
\item $\relentr_D(1) = 1$; for $1 < j < d$, $\relentr_D(j) = \relentr_1(j)$; $\relentr_D(d) = \relentr_1(1)$.
\item If $\relentr_D(d-1) = 0$ then $|\perm_D(d-1)| = d-1$ and $\perm_D(d) = -1$,\\
otherwise $|\perm_D(d-1)| = 1$ and $\perm_D(d) = -(d-1)$.
\item $T_D = \relentr_D(d-1) = \relentr_1(d-1)$.
\end{itemize}
\end{lemma}
\begin{proof}
The first item is actually proven in the last line of the proof of Lemma~\ref{lem:developmenteven}, with $i = D-2$. The second and third item follow from straightforward calculations, similar to those of the previous lemmas, where we use the fact that, by the assumptions of this subsection, we have $e_{D-1} = \langle -1\rangle$ and $e_D = \langle-(d-1)\rangle$, and therefore $\relentr_D(1) = 1$ and $\relexit_D(d-1) = 0$.
\end{proof}

\begin{lemma}\label{lem:omega}\leavevmode
$\relexit_D = \relentr_1 \compose \omega$.
\end{lemma}

\begin{proof}
Straightforward rewriting of the equations in Lemma~\ref{lem:finale}, using Observation~\ref{obs:relationrelentryandexit}, yields:
\begin{itemize}
\item $\relexit_D(1) =
  \relentr_D(1) + \big[|\perm_D(d)| = 1\big] \bmod 2 =
  \relentr_D(1) + \big[\relentr_D(d-1) = 0\big] \bmod 2 =
  \relentr_D(1) + \big[\relentr_D(d-1) = 0\big] \bmod 2 =
  1 + (\relentr_1(d-1) + 1) \bmod 2 = \relentr_1(d-1)$;
\item for $1 < i < d-1$, we have $\relexit_D(i) = \relentr_D(i) + \big[|\perm_D(d)| = i\big] \bmod 2 = \relentr_1(i) + 0 = \relentr_1(i)$;
\item by the assumptions of this subsection, $\relexit_D(d-1) = 0$, which equals $\relentr_1(d)$;
\item $\relexit_D(d) = \relentr_D(d) + \big[|\perm_D(d)| = i\big] \bmod 2 = \relentr_1(1) + 0 = \relentr_1(1)$.
\end{itemize}
This establishes $\relexit_D = \relentr_1 \compose \omega$ with $\omega$ as in Definition~\ref{def:omega}.
\end{proof}

\subsection{Possible entry points of self-similar curves}\label{sec:possibleentry}

In this subsection we will first use the similarity between the 2-curves that make up $A_3$ to prove Lemma~\ref{lem:tworelativeentries}, which states that $\relentr_1(j)$ should be the same for all $j \in \{1,...,d-1\}$. After that, we will use the similarity between $A_2$ and the 2-curve that forms the beginning of any approximating curve $A_k$ ($k \geq 2$), to prove Lemma~\ref{lem:sameonalllevels}, which states that $\relentr_k = \relentr_1$ for all $k \geq 1$. From that we will derive Theorem~\ref{thm:twoentries}, which essentially says that for any fixed $d$, there are only two points that may be the starting point of a $d$-dimensional self-similar, hyperorthogonal, well-folded space-filling curve.

Let $A'_1, A'_2, A'_3$ be extended hyperorthogonal well-folded approximating curves of a self-similar space-filling curve, fulfilling the assumptions which we made, without loss of generality, in Section~\ref{sec:entrytoexit}. When we inflate $A'_2$ to obtain $A'_3$, so that a 2-curve replaces each vertex of $A'_1$, the relative coordinates of each 2-curve's exit point should equal the relative coordinates of the next 2-curve's entry point---otherwise the 2-curves would not be connected by an edge.

\begin{observation}\label{obs:(not)reverse}
Because of self-similarity, the 2-curve replacing $v_i$ of $A_1$ must itself be an non-reverse isometry of either $A_2$ if $T_i = 0$, or $\reverse{A_2}$ if $T_i = 1$.
\end{observation}
Note the either-or in the above observation: by Lemma~\ref{lem:asymmetric}, the 2-curve replacing $v_i$ cannot be a non-reverse isometry of both $A_2$ and $\reverse{A_2}$ at the same time.

As a result of the transformation $\perm_{i-1}$, the relative coordinates of the exit point of the 2-curve replacing $v_{i-1}$ of $A_1$ are given by the function $\relentr_1 \compose \omega \compose |\perm^{-1}_{i-1}|$ if $T_{i-1} = 0$, and by $\relentr_1 \compose |\perm^{-1}_{i-1}|$ if $T_{i-1} = 1$. The relative coordinates of the entry point of the 2-curve replacing $v_{i}$ are given by the function $\relentr_1 \compose |\perm^{-1}_i|$ if $T_i = 0$, and by $\relentr_1 \compose \omega \compose |\perm^{-1}_i|$ if $T_i = 1$. Thus we get:

\begin{lemma}\label{lem:matchinglocations}\leavevmode
\begin{itemize}
\item If $T_{i-1} = 0$ and $T_i = 0$, we have $\relentr_1 \compose \omega \compose |\perm^{-1}_{i-1}| = \relentr_1 \compose |\perm^{-1}_i|$
\item If $T_{i-1} = 0$ and $T_i = 1$, we have $\relentr_1 \compose \omega \compose |\perm^{-1}_{i-1}| = \relentr_1 \compose \omega \compose |\perm^{-1}_i|$
\item If $T_{i-1} = 1$ and $T_i = 0$, we have $\relentr_1 \compose |\perm^{-1}_{i-1}| = \relentr_1 \compose |\perm^{-1}_i|$
\item If $T_{i-1} = 1$ and $T_i = 1$, we have $\relentr_1 \compose |\perm^{-1}_{i-1}| = \relentr_1 \compose \omega \compose |\perm^{-1}_i|$
\end{itemize}
\end{lemma}

\noindent We will now analyse the possible successions of types $T_i$ and permutations $\perm_i$ for the vertices $v_i$ of $A_1$, where $i \in \{1,\ldots,2^d\}$. We will do so in four lemmas, concluding with Lemma~\ref{lem:tworelativeentries}, which states that $\relentr_1(j)$ should be the same for all $j \in \{1,...,d-1\}$.

\begin{lemma}\label{lem:noalternation}
There is a $j \in \{2,3,\ldots,D\}$ such that $T_j = T_{j-1}$.
\end{lemma}
\begin{proof}
There is an even $i < D = 2^d$ (specifically, we may choose $i = D/4$ or $i = 3D/4$) such that $|e_i| = d-1$. By Lemma~\ref{lem:developmenteven}, we have $T_{i+1} = \relentr_1(d-1)$, which, by Lemma~\ref{lem:finale}, equals $T_D$. Thus the sequence $T_{i+1},T_{i+2},\ldots,T_D$ consists of an even number of types where the last equals the first. This implies that a strictly alternating type sequence is not possible.
\end{proof}

\begin{lemma}\label{lem:1equalsd-1}
$\relentr_1(1) = \relentr_1(d-1)$.
\end{lemma}
\begin{proof}
Let $i$ be the largest $i$ from $\{2,3,\ldots,D\}$ such that $T_i = T_{i-1}$ (there is always such an $i$, by Lemma~\ref{lem:noalternation}). We distinguish three cases: (i) $|\perm_{i-1}(1)| = |\perm_i(1)|$; (ii) $|\perm_{i-1}(1)| \neq |\perm_i(1)|$ and $d \geq 4$; (iii) $|\perm_{i-1}(1)| \neq |\perm_i(1)|$ and $d = 3$.

In the first case, let $x$ be $\perm_i(1)$, so we have $|\perm^{-1}_{i-1}(x)| = |\perm^{-1}_i(x)| = 1$. Then, from Lemma~\ref{lem:matchinglocations}, evaluating the functions on both sides for $x$, we find, both in the case of $T_i = T_{i-1} = 0$ and the case of $T_i = T_{i-1} = 1$, the following: $\relentr_1(1) = \relentr_1(\omega(1)) = \relentr_1(d-1)$.

In the second case, we have, by Theorem~\ref{thm:hyperorthogonal}, $|\perm_{i-1}(1)| = |\perm_i(2)|$ and $|\perm_{i-1}(2)| = |\perm_i(1)|$. Let $x=|\perm_i(1)|; y=|\perm_i(2)|$, so we have $|\perm^{-1}_{i-1}(y)| = |\perm^{-1}_i(x)| = 1$ and $|\perm^{-1}_{i-1}(x)| = |\perm^{-1}_i(y)| = 2$.
Since $d\geq 4$ we have $\omega(2)=2$. Now, if $T_{i-1} = T_i = 0$, Lemma~\ref{lem:matchinglocations} gives us $\relentr_1(\omega(|\perm^{-1}_{i-1}(x)|)) = \relentr_1(|\perm^{-1}_i(x)|) \leftrightarrow \relentr_1(2) = \relentr_1(1)$ and
$\relentr_1(\omega(|\perm^{-1}_{i-1}(y)|)) = \relentr_1(|\perm^{-1}_i(y)|) \leftrightarrow \relentr_1(d-1) = \relentr_1(2)$, so $\relentr_1(1) = \relentr_1(2) = \relentr_1(d-1)$.
Otherwise we must have $T_{i-1} = T_i = 1$ and Lemma~\ref{lem:matchinglocations} gives us $\relentr_1(|\perm^{-1}_{i-1}(x)|) = \relentr_1(\omega(|\perm^{-1}_i(x)|)) \leftrightarrow \relentr_1(2) = \relentr_1(d-1)$ and
$\relentr_1(|\perm^{-1}_{i-1}(y)|) = \relentr_1(\omega(|\perm^{-1}_i(y)|)) \leftrightarrow \relentr_1(1) = \relentr_1(2)$, so, again, $\relentr_1(1) = \relentr_1(2) = \relentr_1(d-1)$.

The third case does not occur, since for $d = 3$, the proof of Lemma~\ref{lem:noalternation} yields $T_8 = T_7$, so $i = 8$. Since $G'(3)$ ends with $\langle \ldots, -2, -1, -2\rangle$, both $v_7$ and $v_8$ are incident on edges with axes $1$ and $2$, and both $\perm_7$ and $\perm_8$ must have the remaining axis, 3, at depth 1, thus $|\perm_{i-1}(1)| = |\perm_i(1)| = 3$.
\end{proof}

\begin{lemma}\label{lem:2equalsd-2}
$\relentr_1(j) = \relentr_1(j-1)$ for all $j \in \{3,4,\ldots,d-2\}$.
\end{lemma}
\begin{proof}
By Theorem~\ref{thm:hyperorthogonal}, we must have $\permdp(\perm_h,x) = 0$ for some $h$ and $x = |\perm_1(1)|$, so $\permdp(\perm_1,x) = d-2$. Since depth differs by at most one between successive permutations $\perm_i$, there must be, for any $j \in \{2,\ldots,d-2\}$, an $i$ such that $|\perm_i(j)| = |\perm_{i-1}(j-1)| = x$. Note also that for $j \in \{2,\ldots,d-2\}$, we have $\omega(j) = j$. Hence, from Lemma~\ref{lem:matchinglocations}, evaluating the functions on both sides for $x$, we find $\relentr_1(j-1) = \relentr_1(j)$.
\end{proof}

\begin{lemma}\label{lem:tworelativeentries}
$\relentr_1(j) = \relentr_1(j-1)$ for all $j \in \{2,\ldots,d-1\}$.
\end{lemma}
\begin{proof}
If $d = 3$, the lemma is equivalent to Lemma~\ref{lem:1equalsd-1}. Otherwise, choose $i$ and $x$ such that $|\perm_i(2)| = |\perm_{i-1}(1)| = x$ (such $i$ and $x$ exist, as observed in the proof of Lemma~\ref{lem:2equalsd-2}). Now, if $T_i = T_{i-1}$, the proof of the second case of Lemma~\ref{lem:1equalsd-1} tells us that $\relentr_1(1) = \relentr_1(2) = \relentr_1(d-1)$, and Lemma~\ref{lem:tworelativeentries} follows by combining this fact with Lemma~\ref{lem:2equalsd-2}. It remains to discuss the cases in which $T_i \neq T_{i-1}$.

If $T_{i-1} = 0$ and $T_i = 1$, by Lemma~\ref{lem:matchinglocations}, we have $\relentr_1(d-1) = \relentr_1(\omega(1)) = \relentr_1(\omega(\perm^{-1}_{i-1}(x))) = \relentr_1(\omega(\perm^{-1}_i(x))) = \relentr_1(\omega(2)) = \relentr_1(2)$. Combining this with Lemmas \ref{lem:1equalsd-1} and~\ref{lem:2equalsd-2} establishes $\relentr_1(j) = \relentr_1(j-1)$ for all $j \in \{2,3,\ldots,d-1\}$.

If $T_{i-1} = 1$ and $T_i = 0$, by Lemma~\ref{lem:matchinglocations}, we have $\relentr_1(1) = \relentr_1(\perm^{-1}_{i-1}(x)) = \relentr_1(\perm^{-1}_i(x)) = \relentr_1(2)$. Combining this with Lemmas \ref{lem:1equalsd-1} and~\ref{lem:2equalsd-2} establishes $\relentr_1(j) = \relentr_1(j-1)$ for all $j \in \{2,3,\ldots,d-1\}$.
\end{proof}

We can now use the similarity between $A'_2$ and the 2-curve that forms the beginning of any approximating curve $A'_k$, to prove the following:

\begin{lemma}\label{lem:sameonalllevels}
$\relentr_{k,1} = \relentr_{1,1}$ for all $k \geq 1$.
\end{lemma}

\begin{proof}
For $k=1$ the Lemma is trivial. Now consider the case $k \geq 2$. By self-similarity, $A'_k$ starts with a non-reverse isometry of either $A'_2$ or $\reverse{A'_2}$. In the first case we have $\relentr_{k,1} = \relentr_{1,1} \compose |\perm^{-1}_{k,1}|$; in the second case we have $\relentr_{k,1} = \relexit_{1,D} = \relentr_{1,1} \compose \omega \compose |\perm^{-1}_{k,1}|$. In either case, $\relentr_{k,1}(1,\ldots,d)$ is a permutation of $\relentr_{1,1}(1,\ldots,d)$, which, by Lemma~\ref{lem:tworelativeentries}, can have only two values: it is either all zeros, or $\relentr_{1,1}(d) = 0$ and otherwise it is all ones. In the first case, any permutation is without effect so $\relentr_{k,1} = \relentr_{1,1}$ for any $k$. In the second case, we must have $\relentr_{k,1}(d) = 0$ for any $k$ because $e_{k,0} = d$, and it follows that $\relentr_{k,1}(1,\ldots,d-1)$ is all ones for any $k$.
\end{proof}

\begin{lemma}\label{lem:twoentries}
The combination of Lemmas \ref{lem:tworelativeentries} and~\ref{lem:sameonalllevels} is equivalent to:
\begin{itemize}
\item $\flipped(\sigma^{-1}_{k,1}(j)) = 0$ for all $k$ and all $j$; or
\item $\flipped(\sigma^{-1}_{k,1}(j)) = 0$ if $k$ is even or $j = d$, and $\flipped(\sigma^{-1}_{k,1}(j)) = 1$ if $k$ is odd and $j < d$.
\end{itemize}
\end{lemma}
\begin{proof}
Recall Observation~\ref{obs:relationabsandrel}: $\relentr_{k,1}(j) = (\absentr_{k,1}(j) + v_{k,1}[j]) \bmod 2$,
where $\absentr_{k,1}(j) = \flipped(\sigma^{-1}_{k,1}(j)) \pmod 2$.
Therefore, $\flipped(\sigma^{-1}_{k,1}(j)) = \absentr_{k,1}(j) = \relentr_{k,1}(j) + v_{k,1}[j] \pmod 2$.
Since the entry point of $C_{k,1}$ is, by definition, $v_{k+1,1}$, we obtain
$\flipped(\sigma^{-1}_{k+1,1}(j)) = \relentr_{k+1,1}(j) + v_{k+1,1}[j] = \relentr_{k+1,1}(j) + \flipped(\sigma^{-1}_{k,1}(j)) \pmod 2$, and therefore
$\sgn(\sigma^{-1}_{k+1,1}(j)) = \sgn(\sigma^{-1}_{k,1}(j))$ if and only if $\relentr_{k+1,1}(j) = 0$.

The lemma now follows by straightforward induction from the base case $k = 0$ (in which case $\sigma^{-1}_{k,1}$ is the identity permutation) and the possible values of $\relentr_{k,1}$ as given by Lemmas \ref{lem:tworelativeentries} and~\ref{lem:sameonalllevels}.
\end{proof}

\begin{theorem}\label{thm:twoentries}
If $f$ is a self-similar, hyperorthogonal, well-folded, space-filling curve mapping $[0,1]$ to $[0,1]^d$, then, modulo reflection, reversal and rotation, the entry $f(0)$ is either $(0,\ldots,0,0)$ or $(\frac13,\ldots,\frac13,0)$.
\end{theorem}

\begin{proof}
This is a direct translation of Lemma~\ref{lem:twoentries} using Observation~\ref{obs:startingpoint}.
\end{proof}

\subsection{Construction of self-similar curves}\label{sec:selfsimilarconstruction}

We will now show that curves with the entry points that may exist according to Theorem~\ref{thm:twoentries} do indeed exist for any $d \geq 3$:

\begin{theorem}\label{thm:selfsimilarexistence}
For any $d \geq 3$, there is a self-similar, hyperorthogonal, well-folded $d$-dimensional space-filling curve starting at $(0,\ldots,0,0)$ and there is a self-similar, hyperorthogonal, well-folded $d$-dimensional space-filling curve starting at $(\frac13,\ldots,\frac13,0)$.
\end{theorem}

\begin{proof}
It suffices to show that the construction of Theorem~\ref{thm:existence}, with entry direction $\langle d\rangle$, exit direction $\langle -(d-1)\rangle$, and signs of $\perm^{-1}_{k,1}$ corresponding to either $(0,\ldots,0,0)$ or $(\frac13,\ldots,\frac13,0)$, results in a self-similar curve.

Let $x$ be $\relentr_{1,1}(1)$. Applying the translation of Lemma~\ref{lem:twoentries} and Theorem~\ref{thm:twoentries} in the other direction, we find that both starting points satisfy Lemma~\ref{lem:tworelativeentries} and Lemma~\ref{lem:sameonalllevels}, so the relative entry coordinates of the first child curve $C_{k,1}$ on any level $k$ are given by $\relentr_{k,1}(j) = x$ for $1 \leq j \leq d-1$, and $\relentr_{k,1}(d) = 0$.

By Lemma~\ref{lem:types}, all child curves of the constructed approximating curves $A'_0,A'_1,...$ are an image of an isometry of $G'(d)$. By Lemma~\ref{lem:omega}, the relative coordinates of the entries and exits of the one-step inflation of each such child curve are permutations of each other, and by Observation~\ref{obs:relationrelentryandexit}, a child curve's relative entry coordinates are the previous child curve's relative exit coordinates. Thus, the relative entry and exit coordinates of the one-step inflations of \emph{all} child curves are permutations of $(x,\ldots,x,0)$. Because Theorem~\ref{thm:existence} guarantees the continuity of the approximating curves, we have $\relentr_{k,i}(j) = 0$ if $|e_{k,i-1}| = j$ and $\relentr_{k,i}(j) = x$ if $|e_{k,i-1}| \neq j$; similarly, we have $\relexit_{k,i}(j) = 0$ if $|e_{k,i}| = j$ and $\relexit_{k,i}(j) = x$ if $|e_{k,i}| \neq j$.

Thus, any extended child curve's inflation, to any depth of recursion, consists of child curves of type 0 and 1, with the entry and exit points determined by the fact that all relative entry and exit coordinates are equal to $x$, except that we have $\relentr_i\big(|e_{i-1}|\big) = 0$ and $\relexit_i\big(|e_i|\big) = 0$. Thus, the entry point (or, in the case of reversal, the exit point) of the inflation of any vertex $v_{1,i}$ to a depth of $k$ levels is completely  determined by $\perm_{1,i}$ in the same way in which the entry point of $A_k$ is determined by $\perm_{0,1}$ (which is the identity permutation). As a result, the inflation of $v_{1,i}$ must be a translation of $\perm_{1,i}(A_k)$ or its reverse; hence the space-filling curve is self-similar.
\end{proof}

\noindent It turns out that there are actually very few such curves for $d = 3$ and $d = 4$:
\begin{observation}\label{obs:nopermchoice}
If $d = 3$ or $d = 4$, Lemma~\ref{lem:edgedistance} leaves no choice with respect to the last two elements, the third-last element, and the first element of the permutations $|\perm_{k,i}|$ in a self-similar curve.
\end{observation}

\begin{proof}
By Lemma~\ref{lem:selfsimilartypes}, we may assume that the space-filling curve is approximated by extended hyperorthogonal well-folded curves $A'_0,A'_1,\ldots$ with the entry and exit direction fixed at $\langle d\rangle$ and $\langle -(d-1)\rangle$, respectively.

The last two elements of any permutation $|\perm_{k,i}|$ must be the two different axes of the edges incident on $v_{k,i}$.

If $d = 3$, the third-last (and first) element must be the only remaining axis.

If $d > 3$, the third-last element must be the third axis that is within edge distance 1 from $v_{k,i}$. For $i = 1$, this third axis is $|e_{k,2}|$, which must differ from $|e_{k,0}|$ and $|e_{k,1}| = |e_{k,3}|$, otherwise $\langle e_{k,0}, e_{k,1}, e_{k,2}, e_{k,3}\rangle$ would constitute a sequence of four edges with only two different axes, contradicting Definition~\ref{def:hyperorthogonal}. By a symmetric argument, for $i = \twodk$, the third axis is $|e_{k,i-2}|$. For $1 < i < \twodk$, a third axis must also exist, otherwise the two edges preceding $v_{k,i}$ and the two edges following $v_{k,i}$ would constitute a sequence of four edges with only two axes. If $d = 4$, with the last two elements and the third-last element fixed, the first element must be the only remaining axis.
\end{proof}


\begin{figure}[t]
\centering
\hbox to \hsize{%
\includegraphics[scale=0.25]{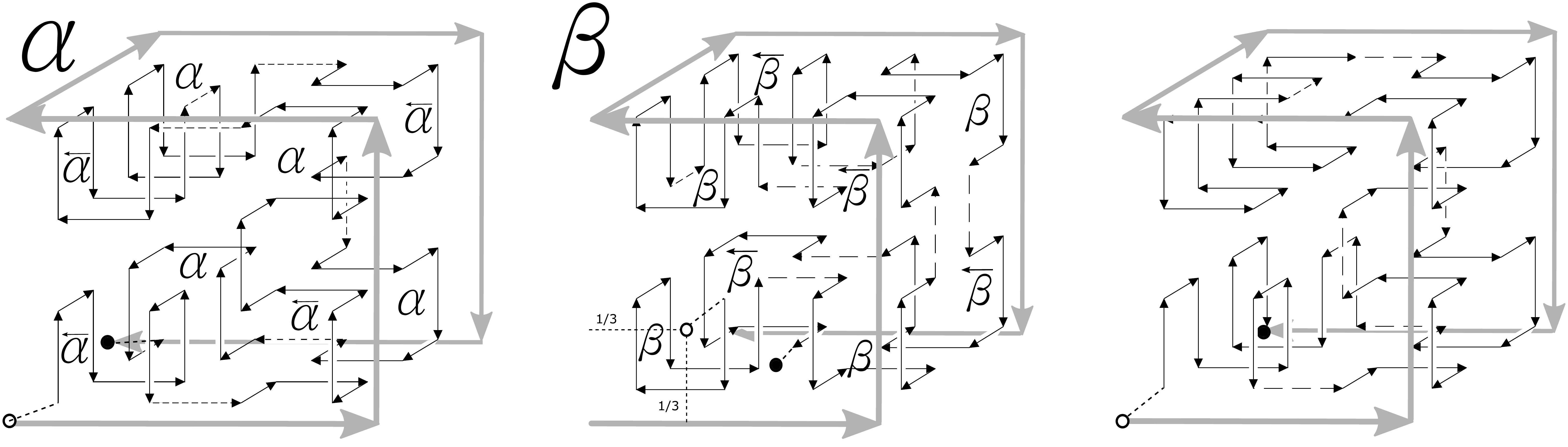}
}\caption{The three-dimensional, self-similar, hyperorthogonal, well-folded space-filling curve with starting points $(0,0,0)$ ($\alpha$, left) and $(\frac13,\frac13,0)$ ($\beta$, centre), and the three-dimensional curve by Butz and Moore (right). The bold grey curve shows $A_1$. The solid black curves depict the child curves of $A_1$, the dashed lines between them indicate how they are connected. The symbols next to the child curves indicate whether they are reversed, with arrow, or not, without arrow. For the Butz-Moore curve, no such indications are given, because the curve is symmetric and there is no need to distinguish between reflections and reversals. The white and black dots indicate the location of the entry $f(0)$ and the exit $f(1)$.}
\vspace{-.2cm}
\label{fig:3dcurves}
\end{figure}

\begin{corollary}\label{cor:onlytwoselfsimilar}
If $d = 3$ or $d = 4$, there are exactly two self-similar, hyperorthogonal, well-folded $d$-dimensional space-filling curves.
\end{corollary}
\begin{proof}
For self-similar curves, by Lemma~\ref{lem:selfsimilartypes}, we may assume the entry and exit direction to be fixed at $\langle d\rangle$ and $\langle -(d-1)\rangle$, respectively. For the starting point, that is, the signs of $\perm^{-1}_{k,1}(j)$ for all $k$ and $j$, only two combinations are possible (Theorem~\ref{thm:twoentries}). Theorem~\ref{thm:existence} states that this leads to two unique hyperorthogonal, well-folded space-filling curves in which the elements of each $|\perm_{k,i}|$ are sorted by order of decreasing local edge distance to $v_{k,i}$ in $A'_k$. By Observation~\ref{obs:nopermchoice}, for $d = 3$ and $d = 4$, there is no other way to order the elements of each $|\perm_{k,i}|$.
\end{proof}

The two three-dimensional self-similar, hyperorthogonal, well-folded space-filling curves are illustrated in Figure~\ref{fig:3dcurves}, left ($\alpha$), and centre ($\beta$).

\section{Implementation in software}\label{sec:implementation}

\subsection{Typical operations}

In order to apply hyperorthogonal well-folded space-filling curves in practical applications, one needs to implement one or more operators based on these curves. Recall that the space-filling curves under consideration in this paper are functions $f: [0,1] \rightarrow [0,1]^d$, with approximating curves $A_0, A_1, \ldots$. Common operators for such curves include:\begin{itemize}
\item \emph{discrete index-to-point conversion}: given a resolution parameter $k$ and an index $i \in \{1,...,\twodk\}$, compute the coordinates of vertex $v_i$ of $A_k$;
\item \emph{continuous index-to-point conversion}: given a number $x \in [0,1]$, calculate $f(x)$;
\item \emph{discrete point-to-index conversion}: given a resolution parameter $k$ and the coordinates of a vertex $v$ of $A_k$, compute the index $i$ such that $v_i = v$;
\item \emph{continuous point-to-index conversion}: given a point $p \in [0,1]^d$, calculate $f^{-1}(p)$;
\item \emph{discrete comparison}: given a resolution parameter $k$ and the coordinates of two vertices $u, v$ of $A_k$, compute which of the two appears earlier along $A_k$;
\item \emph{continuous comparison}: given the coordinates of two points $p, q \in [0,1]^d$, compute which of the two appears earlier along the curve $f$, that is, decide whether $f^{-1}(p) < f^{-1}(q)$ or $f^{-1}(p) > f^{-1}(q)$.
\end{itemize}
There is a catch here: the inverse $f^{-1}$ of a space-filling curve $f$ is not immediately well-defined. For a given point $p$, there may be an approximating curve $A_k$ such that two (or more) hypercubes $H_i$ and $H_j$, corresponding to vertices $v_i$ and $v_j$ on $A_k$, each have $p$ on their boundary, where $j - i > 1$. That implies that there will be a value $x \in [(i-1)/\twodk, i/\twodk]$ and a different value $y \in [(j-1)/\twodk, j/\twodk]$ such that $f(x) = f(y) = p$. A common solution to obtain a unique value for $f^{-1}(p)$ is to ``err on the far side'': for any level $k$, assign each point $p$ to the vertex $v_i$ of $A_k$ whose corresponding hypercube $H_i$ contains the immediate vicinity of $p$ in the direction away from the origin. In other words, we define $f^{-1}(p)$ as the limit of the elements of $\{x \in [0,1] \mid f(x) = p'\}$ as $p'$ approaches $p$ in a straight line directed towards the origin. A drawback of this solution is that $f^{-1}(p)$ is undefined when one or more of the coordinates of $p$ are equal to~1. An alternative solution could be to define $f^{-1}(p)$ as the smallest value $x$ such that $f(x) = p$.

In the context of this publication, it would go too far to go into the details of the optimal implementation of each of the operators mentioned above, with various definitions of $f^{-1}$. Fortunately, the implementations of these operators share the same global structure: starting from the unit hypercube, the operator zooms in onto successively smaller hypercubes until the required output can be delivered. In Section~\ref{sec:implementcomparison} we sketch briefly how to implement the continuous comparison operator with the err-on-the-far-side definition of $f^{-1}$ for the $d$-dimensional self-similar hyperorthogonal well-folded space-filling curves that underlie Theorem~\ref{thm:selfsimilarexistence}. Further details are provided in Appendix~\ref{sec:pseudocode}. A good understanding of our implementation should enable the reader to implement any of the other operators.

\subsection{Implementation of a comparison operator}\label{sec:implementcomparison}

It is relatively easy to implement an efficient comparison operator that decides which of any two given points comes first along a $d$-dimensional, self-similar, hyperorthogonal, well-folded space-filling curve. For a fixed choice of space-filling curve $f$, a recursive implementation would take as input two points $p, q \in [0,1)^d$ that need to be compared, along with a signed permutation $\sigma$ that specifies how the given curve is placed in the unit cube, and the direction of the curve (forward or reversed). Let $S(p)$ and $S(q)$ be the subcubes of width $1/2$ that contain $p$ and $q$, respectively.

If $p = q$, one point does not precede the other. Otherwise, if $S(p) \neq S(q)$, one can decide immediately which point comes first, based on the relative order of the vertices that represent $S(p)$ and $S(q)$ along the approximating 1-curve $\sigma(G(d))$. Finally, if $S(p) = S(q)$, that is, $p$ and $q$ lie in the same subcube of width $1/2$, then their relative order can be decided by a recursive call with:\begin{itemize}
\item the points $p$ and $q$, scaled and translated according to the transformation that maps $S(p)$ to the unit cube;
\item the signed permutation and direction that specifies how the space-filling curve traverses~$S(p)$.
\end{itemize}
In fact, thanks to the structure of the approximating curve $\sigma(G(d))$, one can examine the coordinates of $p$ and $q$ one by one, from the coordinate in dimension $|\sigma(d)|$ down to the coordinate in dimension $|\sigma(1)|$: as soon as a coordinate is found in which the binary representations of the fractional parts of $p$ and $q$ differ in the first bit, one can decide which of the two points precedes the other. Only if $p$ and $q$ are equal in the first bits of all coordinates, the algorithm needs to go in recursion.

To be able to make the recursive call, the algorithm needs to determine the permutation to use in recursion, that is, the transformation that maps the complete space-filling curve $f$ to the section within $S(p)$, modulo scaling and translation. For the curves described by the constructions of Lemma~\ref{lem:edgedistalgo} and Theorem~\ref{thm:selfsimilarexistence} this is relatively straightforward. To determine the unsigned permutation to be used in recursion, we sort the $d$ coordinate axes by decreasing local edge distance $S(p)$. This sorted list of axes can be constructed on the fly in $\Theta(d)$ time while examining the $d$ coordinates of $p$ and $q$ to decide in which subcube they lie. By Lemma~\ref{lem:edgedistalgo}, the sorted list of axes gives us the (unsigned) permutation to use in recursion. The signs of the permutation to use in recursion now follow from applying the observations on relative entry points and permutation signs calculated in the previous section. For further details and pseudocode of a (non-recursive) implementation, see the appendix.

If the binary representations of the coordinates of $p$ and $q$ consist of $k$ bits per coordinate, and we can extract these bits in order of
decreasing significance in constant time per bit, then the complete comparison operator runs in $O(d \ast k)$ time.

\section{Evaluation}\label{sec:discussion}

\subsection{Comparing to the Butz-Moore curves}\label{sec:comparison}

The generalization of Hilbert's curve to $d$ dimensions by Butz~\cite{Butz}, as implemented by Moore~\cite{Moore}, is a self-similar well-folded curve with starting point in the origin, in which the orientations (and therefore, the signs of the inverse permutations) of the child curves of $A_1$ are the same as in our hyperorthogonal well-folded curves. Concretely, $|\perm_i(d)| = 1$ for $i \in \{1, 2^d\}$, and $|\perm_i(d)| = \max(|e_{i-1}|,|e_i|)$ for $1 < i < 2^d$. However, otherwise the permutations are different: all permutations in the Butz-Moore curves are rotations (in the permutation sense of the word), so $|\perm_i(j)| = |\perm_i(d)| + j \pmod d$. For a graphical description of the 3-dimensional curve, see Figure~\ref{fig:3dcurves} (right).

\begin{theorem}\label{thm:butz}
The $d$-dimensional Butz-Moore curve contains subcurves with box-to-curve ratio $\Omega(2^{d/2})$.
\end{theorem}
\begin{proof}
Assume $d \geq 3$. Then $G(d)$ contains a sequence $\langle 1, 2, -1, (2 + \lfloor d/2\rfloor), 1\rangle$ or a sequence $\langle 1, -2, -1, (2 + \lfloor d/2\rfloor), 1\rangle$. Hence, for the child curves of $A_1$, there is an $i$ such that $|\perm_i(d)| = 2$, $|e_i| = 1$, and $|\perm_{i+1}(d)| = 2 + \lfloor d/2\rfloor$. Now consider the last $2^{\lfloor d/2\rfloor - 1}$ edges of $C_i$ and the first $2^{\lfloor d/2\rfloor - 1}$ edges of $C_{i+1}$. By Lemma~\ref{lem:prefixspan}, each of these two sets of edges has $\lfloor d/2\rfloor$ different axes. As a result of the rotations $|\perm_i|$ and $|\perm_{i+1}|$, these sets of axes include $\{3,\ldots,2+\lfloor d/2\rfloor\}$ and $\{3+\lfloor d/2\rfloor,\ldots,d,1\}$, respectively, where the latter set reduces to $\{1\}$ if $d < 5$. Together these sets constitute at least the set $\{1,\ldots,d\}\setminus\{2\}$. Thus the curve through the last $2^{\lfloor d/2\rfloor - 1} + 1$ vertices of $C_i$ and the first $2^{\lfloor d/2\rfloor - 1} + 1$ vertices of $C_{i+1}$ has bounding box volume at least $2^{d-1}$, and hence the worst-case box-to-curve ratio is at least $2^{d-1} / (2^{d/2} + 2) = \Omega(2^{d/2})$.
\end{proof}

The worst-case box-to-curve ratio of the Butz-Moore curves is thus in sharp contrast with the worst-case box-to-curve ratio of our hyperorthogonal, well-folded curves, which have \BCR at most 4 for any $d$. For verification we also calculated the actual worst-case \BCR values for $d \in \{2,3,4,5,6\}$ with the software from Sasburg~\cite{Sasburg} (Table~\ref{tab:actualBCR}). Further investigations may be done into average \BCR values over curve sections of a given size, both for the hyperorthogonal and the Butz curves.

It should be noted, however, that \BCR may not be the only relevant measure of bounding-box quality. Haverkort and Van Walderveen~\cite{Haverkort2009} argued that, at least for $d = 2$, the size of the \emph{boundary} of a bounding box may be as important as its volume---although volume and boundary size are usually correlated.
Using Sasburg's software with a generalization of the worst-case bounding box perimeter ratio from Haverkort and Van Walderveen to higher dimensions, we found that by this measure, already for $d = 3$, the self-similar hyperorthogonal well-folded curve with starting point $(\frac13,\frac13,0)$ is better than the Butz curve.

\begin{table}\def\arraystretch{1.4}
\caption{Worst-case box-to-curve ratios for various curves in up to 6 dimensions.}\label{tab:actualBCR}
\begin{tabularx}{\textwidth}{l@{\ }|r@{\ \ }r@{\ \ }r@{\ \ }r@{\ \ }r@{\ \ }r@{}}
\hline
curve & $d = 2$ & $= 3$ & $= 4$ & $= 5$ & $= 6$ & $\geq 7$ \\
\hline
\textit{lower bound face-continuous} & \textit{2.00} & \textit{2.54} & \textit{3.15} & \textit{3.54} & \textit{3.76} & \textit{4--16/(2$^d$+3)}
\\
best claimed non-self-sim. & 2.22\rlap{$^{\mathrm{a}}$} & 2.89\rlap{$^{\mathrm{b}}$} &&&&
\\
self-sim.\ hyp.\ well-fld.\ $f(0) = (0,\ldots,0,0)$ & 2.40\rlap{$^{\mathrm{c}}$} & 3.11 & 3.53 & 3.76 & 3.88 & $\leq 4$
\\
self-sim.\ hyp.\ well-fld.\ $f(0) = (\frac13,\ldots,\frac13,0)$
& & 3.14 & 3.67 & 3.83 & 3.92 & $\leq 4$
\\
\textit{lower bound non-face-continuous}
& \textit{3.00} & \textit{3.50} & \textit{3.75} & \textit{3.87} & \textit{3.93} & \textit{4--4/2$^d$}
\\
Butz-Moore
& 2.40\rlap{$^{\mathrm{c}}$} & 3.11 & 4.74 & 7.08 & 10.65 & $\Omega(2^{d/2})$
\\
\hline
\multicolumn{7}{@{}l}{$^{\mathrm{a}}$\rule[0mm]{0mm}{5mm} $\beta\Omega$-curve~\cite{Wierum} analysed by H\&vW~\cite{Haverkort2009};\quad $^{\mathrm{b}}$ Iupiter~\cite{Haverkort3D};\quad $^{\mathrm{c}}$ Hilbert's curve~\cite{Hilbert}}
\\
\end{tabularx}\vspace{-0.5cm}
\end{table}

\subsection{Lower bounds}\label{sec:lowerbounds}

In this work we study space-filling curves that can be described by a series of approximating curves $A_0,A_1,\ldots,A_n$, where $A_k$ is a curve on the $k$-cube. Within this context, we restricted our search for curves with good worst-case \BCR first to face-continuous curves; then, more specifically, to well-folded curves; then to hyperorthogonal well-folded curves; and finally to self-similar, hyperorthogonal, well-folded curves. We found that if $d = 3$ or $d = 4$, there are only two self-similar hyperorthogonal well-folded space-filling curves. For $d = 5$ and up, there are many more, as Lemma~\ref{lem:edgedistance} then starts to leave room for swaps among the first elements of the permutations $\perm_{k,i}$. We will now address the question of how much room for further improvement there is within these restrictions or if some of these restrictions are dropped.

For $d=2$, Haverkort and Van Walderveen~\cite{Haverkort2009} report that the \BCR of any section of the well-folded, non-self-similar $\beta\Omega$-curve~\cite{Wierum} is 2.22 in the worst case, and for $d = 3$, Haverkort~\cite{Haverkort3D} claims a fairly complicated, non-self-similar, face-continuous curve with a worst-case \BCR of 2.89.
These two constructions, which do not easily generalize to higher dimensions, constitute improvements of less than 10\% with respect to the self-similar hyperorthogonal well-folded curves.

For larger values of $d$, no face-continuous curve can be much better than any hyperorthogonal well-folded curve, since the first is subject to a lower bound that quickly approaches the upper bound of the latter as $d$ grows.
The proof is based on the fact that any such curve must contain a sequence of at most $2^{d-2}+1$ edges that have all axes $\{1,\ldots,d\}$.

\begin{lemma}\label{lem:maxboxsection}
Let $X$ be a $k$-curve constructed by inflation of a single vertex, and let $S$ be a subcurve of $X$.
If $\volume{S}/\volume{X} \geq 2^{d-1} / (2^d - 1)$, then the bounding box of $S$ is the bounding box of $X$.
\end{lemma}
\begin{proof}
The proof goes by induction on increasing values of $k$.

For $k = 0$, we have $\volume{X} = 1$ and $S$ only satisfies $\volume{S}/\volume{X} \geq 2^{d-1} / (2^d - 1)$ if $S = X$, in which case the bounding box of $S$ is indeed the bounding box of $X$.

Now suppose the lemma holds for $(k-1)$-curves, and consider a series of curves $X_0,X_1,\ldots,X_k$ where $X_0$ is a single vertex and each curve $X_i$ ($i > 0$) is constructed by inflating $X_{i-1}$. Let $v_1,\ldots,v_D$ be the vertices of $X_1$, and let $R_i$ be the $(k-1)$-curve within $X_k$ that results from inflating $v_i$.

Let $S$ be a subcurve of $X_k$ with $\volume{S}/\volume{X_k} \geq 2^{d-1} / (2^d - 1)$. Since $2^{d-1} / (2^d - 1) > 1/2$, the curve $S$ consists of, at least, a subcurve $Y$ of a curve $R_y$, a subcurve $Z$ of a curve $R_z$, and the complete curves $R_i$ for $y < i < z$, where $z - y = 2^{d-1}$.

We define $\volume{R} = 2^{d \ast(k-1)}$; note that $\volume{R_i} = \volume{R}$, regardless of $i$. We have $\volume{Y} + \volume{Z} =
\volume{S} - \sum_{i=y+1}^{z-1} \volume{R_i} \geq 2^{d-1} / (2^d - 1) \ast2^d \ast\volume{R} - (z-y-1) \ast\volume{R} = \left(2^d \ast 2^{d-1} / (2^d - 1) - (2^{d-1} - 1)\right) \ast  \volume{R} = \left(1 + 2^{d-1} / (2^d - 1)\right) \ast \volume{R}$.
Hence, since $\volume{Z} \leq \volume{R}$, we have $\volume{Y} \geq 2^{d-1} / (2^d - 1) \ast \volume{R}$ and thus, $\volume{Y} /
\volume{R_y} \geq 2^{d-1} / (2^d - 1)$. By a symmetric argument, $\volume{Z} / \volume{R_z} \geq 2^{d-1} / (2^d - 1)$.

Therefore the bounding box
of $S$ must contain at least the complete bounding boxes of the $(k-1)$-curves $R_i$ for $y \leq i \leq z$. Since $z - y = 2^{d-1}$, the vertices $v_y,\ldots,v_z$ cannot all lie within a $(d-1)$-dimensional 1-cube, so their bounding box must be the full unit cube, and the bounding box of $R_y,\ldots,R_z$ must be the full bounding box of $X$.
\end{proof}

\begin{theorem}\label{thm:lowerbound}
If $f$ is a space-filling curve approximated by a series of curves $A_0,\ldots,A_k$ within the framework of Section~\ref{sec:notation}, then $f$ has a section with BCR at least $4 - 16/(2^d + 3)$.
\end{theorem}

\begin{proof}
Consider the approximating curve $A_1$ with vertices $v_1,\ldots,v_D$ and edges $e_1,\ldots,e_{D-1}$.
Let $z$ be the smallest $z$ such that $\bigcup_{i=1}^z |e_i| = \{1,\ldots,d\}$, and let $y$ be the largest $y < z$ such that $\bigcup_{i=y}^z |e_i| = \{1,\ldots,d\}$. By our choice of $z$, we have $|e_i| \neq |e_z|$ for all $i < z$, and by our choice of $y$, we have $|e_i| \neq |e_y|$ for all $y < i < z$. Hence, all vertices $v_i$ for $y < i \leq z$ must have the same coordinates with respect to dimensions $|e_y|$ and $|e_z|$, and therefore lie within a $(d-2)$-dimensional hypercube of volume $2^{d-2}$, so $z - y \leq 2^{d-2}$. Note that the bounding box of $v_y,\ldots,v_{z+1}$ has volume $2^d$.

For a given $k$, let $R_i$ be the $(k-1)$-curve within $A_k$ that results from inflating $v_{1,i}$. Let $S$ be the subcurve of $A_k$ that starts with the last $\lceil 2^{d(k-1)} \ast 2^{d-1} / (2^d - 1)\rceil$ vertices of $R_y$ and ends with the first $\lceil 2^{d(k-1)} \ast 2^{d-1} / (2^d - 1)\rceil$ vertices of $R_{z+1}$.
We have $\volume{S} < 2^{d(k-1)} \ast (2^{d-2} + 2 \ast 2^{d-1} / (2^d - 1)) + 2$ (the $+2$ results from rounding up). By Lemma~\ref{lem:maxboxsection}, the bounding box of $S$ is the bounding box of the curves $R_y,\ldots,R_{z+1}$, which has volume $2^d \ast 2^{d(k-1)}$.
Hence, the box-to-curve ratio of the section of $f$ corresponding to $S$ is at least
$2^d / (2^{d-2} + 2^d / (2^d - 1) + 2^{1-d(k-1)})$. The limit for $k \to \infty$ is $4 - 16 / (2^d + 3)$.
\end{proof}
For the specific case of $d = 2$, Haverkort and Van Walderveen~\cite{Haverkort2009} prove a stronger lower bound of~2.

Now suppose we drop the restriction to face-continuous curves. More precisely, suppose we have a space-filling curve approximated by a sequence of curves on the grid $A_0,A_1,\ldots$, where we allow our curves on the grid to have diagonal edges, that is, we allow any edge $(v,w)$ such that $w \neq v$ and $|w[j] - v[j]| \leq 1$ for all $j \in \{1,\ldots,d\}$. In that case, the lower bound becomes even worse:

\begin{theorem}\label{thm:diagonallowerbound}
If there is a $k$ and $i$ such that $v_{k,i}$ and $v_{k,i+1}$ differ in at least two coordinates (in other words: if there is a diagonal edge), then $f$ has a section with BCR at least $4 - 4/2^d$.
\end{theorem}
\begin{proof}
Consider the $m$-curves $X$ and $Y$ that replace $v_{k,i}$ and $v_{k,i+1}$ in $A_{k+m}$. Let $S$ be the subcurve of $X$ with volume $\lceil\volume{X} \ast 2^{d-1} / (2^d - 1)\rceil$, ending at the exit point of $X$, and let $T$ be the subcurve of $Y$ with volume $\lceil\volume{Y} \ast 2^{d-1} / (2^d - 1)\rceil$, starting at the entry point of $Y$. By Lemma~\ref{lem:maxboxsection}, the concatenation of $S$, $\langle e_{k,i}\rangle$, and $T$ now has bounding box volume at least $4 \ast 2^{d \ast m}$, while $\volume{S} + \volume{T} \leq 2^{d \ast m} \ast 2^d / (2^d - 1) + 2$. Hence, the box-to-curve ratio of the corresponding section of $f$ is at least $4 / (2^d / (2^d - 1) + 2^{1 - d \ast m})$. The limit for $m \to \infty$ is $4 - 4/2^d$.
\end{proof}

\subsection{Questions for further research}\label{sec:researchquestions}

Note that, as Table~\ref{tab:actualBCR} shows, at least for $d$ up to 6 the lower bound of Theorem~\ref{thm:diagonallowerbound} for curves with ``diagonal edges'' is greater than the worst-case \BCR of the best hyperorthogonal, well-folded curves, and for higher dimensions the difference between the lower bound and the upper bound is less than 1\%. Therefore, in terms of worst-case \BCR, little is to be expected from non-face-continuous curves based on inflation of $k$-cubes for increasing $k$.

The question remains whether there are hyperorthogonal curves that are not well-folded, and if so, whether such curves would also have good bounds on the box-to-curve ratio. In other words: is well-foldedness really required in Theorem~\ref{thm:bcr}? Regardless, Theorem~\ref{thm:lowerbound} shows that in any case, there is not much room for finding curves with a better worst-case \BCR within the framework of Section \ref{sec:notation}.

Can we find space-filling curves with a better worst-case BCR outside this framework? Peano's space-filling curve and its obvious generalization to higher dimensions are based on approximating curves $A_i$ on grids of $3^{d \ast i}$ vertices. For these curves in 2, 3, 4, 5, and 6 dimensions, Sasburg's software~\cite{Sasburg} reports a worst-case BCR of 2.00, 3.06, 3.64, 3.87, and 3.96 respectively. This may serve as evidence that, also for these curves, four is an asymptotic upper bound on the worst-case BCR, regardless of $d$. Note, however, that in higher dimensions, the BCR of these curves seems to be slightly worse than the BCR of our hyperorthogonal well-folded curves.

Departing from the framework of Section~\ref{sec:notation} even further: would it be possible to find space-filling curves with a better worst-case BCR that cannot be approximated by Hamiltonian paths on hypercubic grids? Or are such curves also subject to an asymptotic lower bound of 4?

One may also ask what lower bounds could be proven in more \emph{restricted} settings than that of Section 1.3. For example, Alber and Niedermeier~\cite{Alber} provide a framework for the description of generalizations of Hilbert curves that are self-similar and, in the terminology of Haverkort~\cite{Haverkort3D}, \emph{order-preserving}: $A_{k+1}$ is the concatenation of $2^d$ scaled, translated, rotated and/or reflected \emph{but not reversed} copies of $A_k$. From Lemma~\ref{lem:developmenteven} we know that any approximating curve $A_2$ of a self-similar hyperorthogonal well-folded space-filling curves contains child curves of both types (zero and one), that is, it contains both non-reverse and reverse isometries of $G'(d)$. By Lemma~\ref{lem:asymmetric}, these are really different: no non-reverse isometry of an inflation of $G'(d)$ can visit its vertices in the same order as $\reverse{G'(d)}$. So no self-similar hyperorthogonal well-folded space-filling curves exist without reversal, and thus we get:

\begin{corollary}\label{cor:albernogood}
No $d$-dimensional self-similar hyperorthogonal well-folded curve for $d > 2$ can be described within the framework of Alber and Niedermeier \cite{Alber}.
\end{corollary}

Are the curves that \emph{can} be described within the framework of Alber and Niedermeier subject to an exponential lower bound on the worst-case BCR?

\addcontentsline{toc}{section}{References}
\bibliography{HOWFC}

\clearpage
\appendix

\section{Implementation of a comparison operator}\label{sec:pseudocode}

In this appendix we explain how to implement an efficient comparison operator that decides which of any two given points comes first along a $d$-dimensional self-similar hyperorthogonal well-folded space-filling curve. Algorithm~\ref{alg:comparison} gives an implementation for a curve with entry point $(0,\ldots,0)$, assuming $d \geq 3$. (For $d = 2$, one could use any implementation of Hilbert's curve.) We will briefly explain how the algorithms works below. We have also tested the algorithm and verified that it correctly orders all grid points along hyperorthogonal, well-folded curves, for all grids of $\twodk$ points with $3 \leq d \leq 6$ and $2 \leq k \leq 12/d$.
A truly efficient implementation may call for the use of various bit tricks (for example, an array whose elements are $1$ and $-1$ could be encoded as a single binary number); however, in the interest of readability, with our implementation we strive to stay closer to the theory of Section~\ref{sec:selfsimilar} and avoid tricks that would hide too much of what is going on conceptually.

\SetFuncSty{textsc}
\SetDataSty{textit}
\SetCommentSty{textrm}
\SetKw{True}{true}
\SetKw{False}{false}
\DontPrintSemicolon
\SetKw{KwDown}{down}
\SetKwData{direction}{direction}
\SetKwData{pintheback}{pInTheBack}
\SetKwData{qintheback}{qInTheBack}
\SetKwData{absparentperm}{unsgnedPrm}
\SetKwData{signinvparentperm}{sgnsInvPrm}
\SetKwData{abschildperm}{unsgnedChldPrm}
\SetKwData{signinvchildperm}{sgnsInvChldPrm}
\SetKwData{subcubeindex}{sbcubeId}
\SetKwData{lowerbound}{lowbnd}
\SetKwData{upperbound}{uppbnd}
\SetKwData{entryaxis}{entrAxs}
\SetKwData{exitaxis}{extAxs}
\SetKwData{theaxis}{axis}
\SetKwData{nextaxis}{quartAxs}
\SetKwData{orientation}{orientation}

\setlength\algomargin{1.5em}
\SetAlCapHSkip{1.5em}
\begin{algorithm}
\KwIn{Points $p = (p[1],\ldots,p[d])$ and $q = (q[1],\ldots,q[d])$ in $[0,1)^d$}
\KwOut{$-1$, $0$, or $1$: if $1$, $p$ precedes $q$ along the curve; if $0$, $p = q$; if $-1$, $p$ follows $q$}
\BlankLine
$\direction \gets 1$;
$\absparentperm[0,\ldots,d] \gets [0,\ldots,d]$;
$\signinvparentperm[0,\ldots,d] \gets [1,\ldots,1]$\label{algln:initialization}\;
\Repeat{$p = q$}{\label{algln:mainloop}
  $\entryaxis \gets \absparentperm[d]$;
  $\exitaxis \gets \absparentperm[d-1]$\label{algln:initentryexit}\;
  $\nextaxis \gets \absparentperm[d]$;
  $\subcubeindex \gets 0$\;
  \BlankLine
  \For{$i \gets 1$ \KwTo $d$}{\label{algln:innerloop}
    $\theaxis \gets \nextaxis; \nextaxis \gets \absparentperm[d - i]$;\label{algln:quartaxis}
    $\subcubeindex \gets 2 \cdot \subcubeindex$\label{algln:doubleindex}\;
    \BlankLine
    \tcp{figure out in which half of the cube $p$ and $q$ are:}
    $p[\theaxis] \gets 2 \cdot p[\theaxis]$;
    $\pintheback \gets \lfloor p[\theaxis]\rfloor$;
    $p[\theaxis] \gets p[\theaxis] \bmod 1$\label{algln:extractp}\;
    $q[\theaxis] \gets 2 \cdot q[\theaxis]$;
    $\qintheback \gets \lfloor q[\theaxis]\rfloor$;
    $q[\theaxis] \gets q[\theaxis] \bmod 1$\label{algln:extractq}\;
    \If{$\pintheback \neq \qintheback$}{
      \tcp{on different sides: return $1$ if $p$ comes first; $-1$ if $q$ comes first}
      \Return $\direction \cdot \signinvparentperm[\theaxis] \cdot \sgn(\qintheback - \pintheback)$\label{algln:output}
    }
    \BlankLine
    \tcp{determine sign such that entry point lies on outside:}
    $\signinvchildperm[\theaxis] \gets 1 - 2 \cdot \pintheback$\;\label{algln:setsign}
    \BlankLine
    \If{$\pintheback = \flipped(\signinvparentperm[\theaxis])$}{\label{algln:firstpart}
      $\abschildperm[i-2] \gets \exitaxis$;
      $\exitaxis \gets \theaxis$\label{algln:updateexit}
      \tcp*[r]{$p$ and $q$ in 1st half}
    }
    \Else{
      $\abschildperm[i-2] \gets \entryaxis$;
      $\entryaxis \gets \theaxis$\label{algln:updateentry}\tcp*[r]{$p$ and $q$ in 2nd half}
      $\subcubeindex \gets \subcubeindex + 1$\label{algln:incrementindex}\;
      $\signinvparentperm[\nextaxis] \gets -\signinvparentperm[\nextaxis]$\;\label{algln:reflect}\label{algln:endinnerloop}
    }
  }
  \BlankLine
  \tcp{fill in last two elements of unsigned permutation:}
  $\abschildperm[d-1] \gets \absparentperm[1]$\;\label{algln:setpermd-1}\label{algln:begincorrections}
  $\abschildperm[d] \gets \entryaxis + \exitaxis - \absparentperm[1]$\tcp*[r]{the other axis}\label{algln:setpermd}
  \tcp{in first and last subcube it is the other way around:}
  \lIf{$\subcubeindex \in \{0, 2^d - 1\}$}{swap $\abschildperm[d-1], \abschildperm[d]$}\label{algln:setpermatends}
  \tcp{correct first element of permutation in last quarter:}
  \lIf{$\subcubeindex \geq \frac34\cdot 2^d$}{$\abschildperm[1] \gets \absparentperm[d]$}\label{algln:lastquarter}
  \BlankLine
  \tcp{correct entry point to be on inside w.r.t. $\absparentperm[1]$:}
  $\signinvchildperm[\absparentperm[1]] \gets -\signinvchildperm[\absparentperm[1]]$\;\label{algln:correctsign1}
  \tcp{correct entry point to be on inside w.r.t. orientation of subcube:}
  $\orientation \gets \abschildperm[d]$;
  \If{$\subcubeindex \notin \{0, 2^d-1\}$\label{algln:whentocorrectsignd}}{%
    $\signinvchildperm[\orientation] \gets -\signinvchildperm[\orientation]$\label{algln:correctsignd}\label{algln:endcorrections}
  }
  \BlankLine
  $\absparentperm = \abschildperm$;
  $\signinvparentperm = \signinvchildperm$\;
  \tcp{if type 1, reverse direction:}
  \lIf{$\exitaxis = \orientation$}{$\direction \gets -\direction$}\label{algln:setdirection}
}
\Return 0\tcp*[r]{$p$ and $q$ are equal}
\caption{Comparison operator based on the $d$-dimensional self-similar hyperorthogonal well-folded space-filling curve with entry point $(0,\ldots,0)$, $d \geq 3$.\label{alg:comparison}%
}
\end{algorithm}

\subsection{Input and output of the \textbf{repeat} and \textbf{for} loops}

We will first describe the input and output of the \textbf{repeat} and \textbf{for} loops. After that we will explain how this functionality is implemented.

The code takes two points $p, q \in [0,1)^d$ that need to be compared. In the \textbf{for} loop (Lines \ref{algln:innerloop} to~\ref{algln:endinnerloop}) the algorithm tries to decide which of the two points comes first along the curve, assuming that the curve is reversed as specified by \direction ($1$ means: forward, not reversed; $-1$ means: reversed), and rotated and reflected according to the signed permutation $\sigma$ specified by $\absparentperm$ and $\signinvparentperm$. Here $\absparentperm[1,\ldots,d]$ gives the absolute values of $\sigma(1),\ldots,\sigma(d)$ and $\signinvparentperm[1,\ldots,d]$ holds the signs of $\sigma^{-1}(1),\ldots,\sigma^{-1}(d)$ (the entries $\absparentperm[0]$ and $\signinvparentperm[0]$ are sentinels that are used to prevent indexing arrays out of bounds on Lines \ref{algln:quartaxis} and~\ref{algln:reflect} when $i = d$). On Line~\ref{algln:initialization}, the direction is initialized to forward and $\sigma$ is initialized to the identity permutation.

If $p$ and $q$ lie in the same subcube $H$ of width $1/2$, the \textbf{for} loop ends without returning a result, but as a side effect, it will have done the following:\begin{itemize}
\item $p$ and $q$ are scaled and translated according to the transformation that maps $H$ to the unit cube;
\item the signed permutation that specifies how the curve traverses $H$ has been determined and stored in \abschildperm and \signinvchildperm (modulo some small ``mistakes'', which will be corrected in Lines \ref{algln:begincorrections} to~\ref{algln:endcorrections});
\item the position of $H$ in the order in which the curve traverses the unit cube has been stored in \subcubeindex (0 for the first subcube; $2^d-1$ for the last subcube).
\end{itemize}
The algorithm will then, on Lines \ref{algln:begincorrections} to~\ref{algln:setdirection}, correct the ``mistakes'' and set up \absparentperm, \signinvparentperm and \direction for the next iteration, which effectively zooms in on the subcube~$H$ that contains $p$ and $q$. If, eventually, $p$ and $q$ cannot be distinguished, the algorithm returns~0.

\subsection{Deciding in which subcube $p$ and $q$ lie}

We will now describe how the \textbf{for} loop determines in which subcube(s) $p$ and $q$ lie. For now, the reader may ignore the assignments to \entryaxis, \exitaxis, \abschildperm: these have a role in determining the signed permutation that specifies how the curve traverses the common subcube (if any) of $p$ and $q$; we will get back to that in Section~\ref{sec:implementationsortbyed}.

Recall that the space-filling curve that fills the unit cube is approximated by a curve $\sigma(G(d))$ with vertices $v_1,\ldots,v_D$. Each vertex $v_j$ corresponds to a hypercube of width $1/2$, and in particular, there will be two indices $j$ and $k$ such that $v_j$ and $v_k$ correspond to the hypercubes $H_j$ and $H_k$ that contain $p$ and $q$, respectively. Note that, thanks to our decision to ``err on the far side'', for any $m$ and for any $i \in \{1,...,d\}$, the first bit of the fractional part of coordinate $i$ of any point in $H_m$ is equal to $v_m[i]$. The main goal of the \textbf{for} loop is to identify whether $j < k$ or $j > k$, and, if $j = k$, what is their value.

To this end the \textbf{for} loop implicitly maintains a lower bound \lowerbound and an upperbound \upperbound on $j$ and $k$. In successive iterations, the gap between these bounds is narrowed until we either find $j \neq k$, or $\lowerbound = j = k = \upperbound$. In the last case, $\subcubeindex$ eventually holds the value of $j - 1$ (this is because this article generally indexes vertices starting from one, but the implementation starts from zero). Specifically, the following invariant is valid just before each execution of Line~\ref{algln:quartaxis}: $\lowerbound = 2^{d+1-i} \ast \subcubeindex + 1 \leq j \leq k \leq 2^{d+1-i} \ast (\subcubeindex + 1) = \upperbound$. Just after Line~\ref{algln:quartaxis}, the following holds: (i) $\lowerbound = 2^{d-i} \ast \subcubeindex + 1 \leq j \leq k \leq 2^{d-i} \ast (\subcubeindex + 2) = \upperbound$; (ii) $\theaxis = |\sigma(d+1-i)|$, and (iii) $\nextaxis = |\sigma(d-i)|$.

Note that due to the properties of $G(d)$, we always have that $v_{\lowerbound},\ldots,v_{\upperbound}$ is a translation of $\sigma(G(d+1-i))$ or its reverse, and this curve consists of the concatenation of $\sigma(G(d-i))$, an edge with axis $\sigma(d+1-i)$, and $\sigma(\reverse{G(d-i)})$. In iteration $i$ of the loop, the algorithm decides whether $j$ and $k$ lie in the first or in the second half, that is, before or after the edge with axis $\theaxis = |\sigma(d+1-i)|$. Line~\ref{algln:extractp} reads and removes the first bit of the fractional part of $p[\theaxis]$ (that is, $v_j[\theaxis]$), shifting the remaining bits left for the next iteration of the \textbf{repeat} loop. The bit that is read is stored in \pintheback. Similarly, Line~\ref{algln:extractq} reads and removes the first bit $v_k[\theaxis]$ of the fractional part of $q[\theaxis]$.

If $p$ and $q$ differ in the bits just read, we can now decide, on Line~\ref{algln:output}, which of the two comes first along the space-filling curve. In the absence of any reflections or reversals, we would return 1 if $p$ has the smaller coordinate and $-1$ if $q$ has the smaller coordinate. However, if this portion of the curve is reflected in this coordinate, or if it is reversed, the return value is modified accordingly by multiplying with $\signinvparentperm[\theaxis]$ and \direction.

If, on the other hand, $p$ and $q$ have the same initial bit in dimension $\theaxis$, their shared coordinate is effectively stored in $\signinvchildperm$: if $p$ and $q$ lie in the back (the bits read were ones), we store $-1$; if $p$ and $q$ lie in the front (the bits read were zeros), we store 1. Thus, $\flipped(\signinvchildperm[\theaxis]) = \pintheback = \qintheback$. The next iteration of the \textbf{for} loop must now zoom in onto the isometric copy of $\sigma(G(d-i))$ that contains $v_j$ and $v_k$ out of the two copies that appear before and after the edge with axis $\theaxis = |\sigma(d+1-i)|$ in $\sigma(G(d+1-i))$. There are two possibilities:\begin{itemize}
\item $v_j$ and $v_k$ lie in the first part (traversed according to $\sigma(G(d-i))$): in the absence of reflections, this is the case if $v_k[\theaxis] = 0$, but if $\sigma$ encodes a reflection in dimension $\theaxis$, then $v_j$ and $v_k$ lie in the first part if $v_k[\theaxis] = 1$. Line~\ref{algln:firstpart} checks for this: the condition evaluates to true if and only if $v_j$ and $v_k$ lie in the first part. If so, the upper bound on $j$ and $k$ needs to be lowered. This is realized by incrementing the loop counter $i$ and doubling $\subcubeindex$ on Line~\ref{algln:doubleindex} of the next iteration.
\item $v_j$ and $v_k$ lie in the second part, traversed according to $\reverse{\sigma(G(d-i))}$, or equivalently, $\sigma(G(d-i))$ reflected in coordinate $\sigma(d-i)$. Line~\ref{algln:reflect} implements that reflection, or rather, it ``falsifies'' $\signinvparentperm[\sigma(d-i)]$ such that it tricks the next iteration of the \textbf{for} loop into acting according to a reflection in coordinate $\sigma(d-i)$. Additionally, the lower bound on $j$ and $k$ needs to be raised. This is realized by incrementing \subcubeindex on Line~\ref{algln:incrementindex}, followed by increasing the loop counter $i$ and doubling $\subcubeindex$ on Line~\ref{algln:doubleindex} of the next iteration.
\end{itemize}

Note that, when $i$ becomes $d+1$ so that the \textbf{for} loop terminates, as a result of the loop invariant, if $p$ and $q$ have the same initial bit in each dimension, \subcubeindex will hold the correct value of $j-1$ and $k-1$ once the \textbf{for} loop terminates.

\subsection{Sorting axes by local edge distance}\label{sec:implementationsortbyed}

If the \textbf{for} loop completes, that is, $j = k$, we need to set up the permutation $\sigma$ to use in the next iteration of the \textbf{while} loop. In what follows, we will continue to use $\sigma$ for the permutation used in the current iteration, and we will use $\sigma'$ to denote the permutation to use in the next iteration. We will use the construction of Section~\ref{sec:generalconstruction}: the absolute values of the elements of $\sigma'$ are sorted in order of decreasing local edge distance to $v_j$. To realize this, the algorithm exploits the following property of gray codes:
\begin{observation}\label{obs:sorting}
Let $C$ be the concatenation of $\langle a\rangle$, $G(m)$ and $\langle z\rangle$, where $|a| > m$, $|z| > m$ and $|a| \neq |z|$. If $v$ is a vertex in the first half of $G(m)$, then $\edgedist(C,v,z) > \edgedist(C,v,e)$ for any edge $\langle e\rangle$ in the concatenation of $\langle a\rangle$ and $G(m)$. Symmetrically, if $v$ is a vertex in the second half of $G(m)$, then $\edgedist(C,v,a) > \edgedist(C,v,e)$ for any edge $\langle e\rangle$ in the concatenation of $G(m)$ and $\langle z\rangle$.
\end{observation}
Recall from Section~\ref{sec:extended1curves} that a self-similar hyperorthogonal well-folded curve must have an extended approximating curve $A' = \sigma(A'_1)$ with entry edge $\langle\sigma(d)\rangle$ and exit edge $\langle\sigma(-(d-1)\rangle$. As explained above, in iteration $i$ of the \textbf{for} loop, we are trying to locate $v_j$ and $v_k$ in a section of $A'$ that is isometric to $\sigma(G(d+1-i))$, and we decide whether $v_j$ and $v_k$ appear in the first or in the second half of that curve. As a result, in each iteration of the \textbf{for} loop we may be able to apply Observation~\ref{obs:sorting} to determine one more axis in the sequence of axes sorted by decreasing local edge distance to $v_j$ in $A'$. Considering this idea more carefully, we see that in the first two iterations Observation~\ref{obs:sorting} cannot be applied since some or all of the preconditions $|a| > m$, $|z| > m$, and $|a| \neq |z|$ are violated; this is consistent with the fact that after the last iteration, two axes must remain that both have local edge distance zero and cannot be sorted. In the third iteration, Observation~\ref{obs:sorting} can be applied if $v_j$ is in the first three quarters of $A'$, but things go wrong if $v_j$ is in the last quarter of $A'$, which is a reflection of $\sigma(G(d-2))$ preceded \emph{and} followed by an edge $\langle \sigma(-(d-1))\rangle$ (however, in that case, it is clear that $|\sigma(d)|$ is the axis with the largest edge distance). In each of the iterations after the third we can always determine one more axis in the sorted sequence.

In our implementation, the sorting is implemented by assignments to \abschildperm, supported by assignments to \entryaxis and \exitaxis. For ease of implementation, an axis is assigned to $\abschildperm[i-2]$ in each iteration $i$ of the \textbf{for} loop (on Line \ref{algln:updateexit} or~\ref{algln:updateentry}), but the assignments in the first two iterations (to $\abschildperm[-1]$ and $\abschildperm[0]$) are meaningless and without consequence. Throughout the iterations of the \textbf{for} loop, the algorithm keeps track of the axes \entryaxis and \exitaxis of the edges that precede and follow the curve $\sigma(G(d+1-i))$ currently under consideration, by the assignments on Lines \ref{algln:initentryexit}, \ref{algln:updateexit} and~\ref{algln:updateentry}. Thus, when the \textbf{for} loop ends, \entryaxis and \exitaxis store the two axes at edge distance zero to $v_j$. As noted above, if $v_j$ is in the last quarter a wrong assignment to $\abschildperm[1]$ is made in the third iteration; this is corrected on Line~\ref{algln:lastquarter}.

To complete the permutation $\sigma'$ to use in recursion (modulo the signs), we need to fill in
$\abschildperm[d-1]$ and $\abschildperm[d]$ with the axes $|\sigma(e_{i-1})|$ and $|\sigma(e_i)|$, one of which equals $\sigma(1)$ (by Lemma~\ref{lem:graycodealternatingpattern}). Since the space-filling curve has entry point $(0,\ldots,0)$, we have $\relentr_1(i) = 0$ for all $i \in \{1,\ldots,d\}$. Hence, by Lemmas \ref{lem:opening} to~\ref{lem:finale}, $|\sigma'(d-1)| = \sigma(1)$ and $|\sigma'(d)|$ is the other axis out of $|\sigma(e_{i-1})|$ and $|\sigma(e_i)|$, unless $j \in \{1,2^d\}$ (and thus, $\subcubeindex \in \{0,2^d-1\}$), in which case $|\sigma'(d)| = \sigma(1)$ and $|\sigma'(d-1)|$ is the other axis. Corresponding assignments are made in Lines \ref{algln:setpermd-1} to~\ref{algln:setpermatends}.

\subsection{Computing the signs of the permutation $\sigma'$ of the common subcube of $p$ and $q$}

The correctness of the assignments to $\signinvchildperm$ can be verified by developing the calculations of Lemmas \ref{lem:opening} to~\ref{lem:finale} further. Some waypoints for these calculations are the following. Let $i$ be the rank of the subcube that contains $p$ and $q$ along the curve, that is, $i = \subcubeindex + 1$. As before, we will continue to use $\sigma$ for the permutation used in the current iteration, and we will use $\sigma' = \sigma \compose \sigma_i$ to denote the permutation to use in the next iteration.
The assignments on Line~\ref{algln:setsign} set the signs of $\sigma'$ such that the relative coordinates of the entry point of $C_i$, as defined in Section~\ref{sec:relativecoordinates}, are all zero (recall the relation between subcube coordinates, relative entry coordinates and signs as given by Observation~\ref{obs:relationabsandrel}, which implies $\relentr_i(j) = 0$ if and only if $\flipped\big(\perm^{-1}_i(j)\big) = v_i[j] \bmod 2$). Lines \ref{algln:correctsign1} and~\ref{algln:correctsignd} flip the relative entry coordinates in dimension $\sigma(1)$ and (if $1 < i < 2^d$) dimension $|\sigma'(d)|$. Finally, Line~\ref{algln:setdirection} reverts the direction in all subcubes of type 1, which means that for subcubes of type 1, the aforementioned settings of the relative coordinates are actually for the exit point, not the entry point. We will now explain how the correctness of the settings for subcubes of type 0 can be derived from Lemmas \ref{lem:opening} to~\ref{lem:finale}; the correctness of the settings of the relative exit coordinates for subcubes of type 1 can be derived similarly, using the relations between entry and exit coordinates from Observation~\ref{obs:relationrelentryandexit}.

Note that for the first subcube, all relative entry coordinates are zero. From Lemmas \ref{lem:opening} to~\ref{lem:finale} we get that $C_{i}$ has type 0 if and only if $i = 2^d$, or if $i$ is odd and $i \neq 1$.
The correct settings of the relative entry point coordinates for the case $i = 2^d$ can now be verified directly using Lemma~\ref{lem:finale}. For the case of odd $i \neq 1$, the relative entry coordinates are the relative exit coordinates for $C_{i-1}$, which differ from the relative entry coordinates for $C_{i-1}$ only in the orientation $|\sigma(\sigma_{i-1}(d))|$, which, by Lemma~\ref{lem:developmenteven}, equals $|\sigma(\sigma_{i}(d))| = |\sigma'(d)|$ and differs from $\sigma(1)$. It follows from Lemma~\ref{lem:developmenteven} that the relative entry coordinates for $C_{i}$ are as follows: $\relentr_{i}(\sigma(1)) = \relentr_{i-1}(\sigma(1)) = 1$; for $1 < j < d$ and $j \neq |\sigma'(d)|$ we have $\relentr_{i}(\sigma(j)) = \relentr_{i-1}(\sigma(j)) = \relentr_{1}(\sigma(j)) = 0$; for $j = d$ and $j \neq |\sigma'(d)|$ we have $\relentr_{i}(\sigma(d)) = \relentr_{i-1}(\sigma(d)) = \relentr_{1}(\sigma(1)) = 0$; and for $j = |\sigma'(d)|$ (and hence, $j \neq 1$), we have $\relentr_{i}(\sigma(j)) = 1 - \relentr_{i-1}(\sigma(j)) = 1$. This is exactly what the algorithm establishes.

\subsection{Running time}

If the binary representations of the coordinates of $p$ and $q$ consist of $k$ bits per coordinate, then the algorithm runs in $O(d \cdot k)$ time (that is, linear in the input size), provided that the necessary operations to extract a single bit from a coordinate (Lines~\ref{algln:extractp} and \ref{algln:extractq}) run in constant time.

\subsection{Variant for a space-filling curve with entry point on face}

For a comparison operator based on the $d$-dimensional self-similar hyperorthogonal well-folded space-filling curve with entry point $(\frac13,\ldots,\frac13,0)$, $d \geq 3$, one may adapt Algorithm~\ref{alg:comparison} as follows: flip the sign of the value assigned on Line~\ref{algln:setsign}; swap the values assigned on Lines \ref{algln:setpermd-1} and~\ref{algln:setpermd}; remove Line~\ref{algln:correctsign1}; and change the $\notin$ sign on Line~\ref{algln:whentocorrectsignd} into an $=$ sign. The correctness of these modifications can again be verified with Lemma \ref{lem:opening} to~\ref{lem:finale}, using $\relentr_1(i) = 1$ for all $i \in \{1,\ldots,d-1\}$ and $\relentr_1(d) = 0$.

\end{document}